\documentclass[11pt]{amsart}
\usepackage{amsmath}
\usepackage{graphicx}
\usepackage{bm}
\usepackage{enumerate}
\theoremstyle{plain}
\newtheorem{theorem}{Theorem}[section]
\newtheorem{lemma}[theorem]{Lemma}
\newtheorem{prop}[theorem]{Proposition}

\theoremstyle{definition}

\newtheorem{definition}[theorem]{Definition}

\numberwithin{equation}{section}

\newcommand{\tr}{\operatorname{Tr}}

\newcommand{\abs}[1]{\left|#1\right|}
\begin{document}




\title[Proof of the Mass-Angular Momentum Inequality]{Proof of the Mass-Angular Momentum Inequality for Bi-Axisymmetric Black Holes With Spherical Topology}

\author{Aghil Alaee}
\address{Department of Mathematical and Statistical Science\\
University of Alberta\\
Edmonton AB T6G 2G1, Canada\\
\& Department of Mathematics and Statistics\\
Memorial University of Newfoundland\\
St John's NL A1C 4P5, Canada}
\email{khangha@ualberta.ca, aak818@mun.ca}


\author{Marcus Khuri}
\address{Department of Mathematics, Stony Brook University, Stony Brook, NY 11794, USA}
\email{khuri@math.sunysb.edu}

\author{Hari Kunduri}
\address{Department of Mathematics and Statistics\\
Memorial University of Newfoundland\\
St John's NL A1C 4P5, Canada}
\email{hkkunduri@mun.ca}


\thanks{A. Alaee acknowledges the support of a PIMS Postdoctoral Fellowship and NSERC Grant 261429-2013. M. Khuri acknowledges the support of NSF Grant DMS-1308753. H. Kunduri acknowledges the support of NSERC Grant 418537-2012.}

\begin{abstract}
We show that extreme Myers-Perry initial data realize the unique absolute
minimum of the total mass in a physically relevant (Brill) class of
maximal, asymptotically flat, bi-axisymmetric initial data for the Einstein
equations with fixed angular momenta. As a consequence, we prove the mass-angular momentum inequality in this setting for 5-dimensional spacetimes. That is, all data in this class satisfy the inequality $m^3\geq \frac{27\pi}{32}\left(|\mathcal{J}_1|+|\mathcal{J}_2|\right)^2$, where $m$ and $\mathcal{J}_i$, $i=1,2$ are the total mass and angular momenta of the spacetime. Moreover, equality holds if and only if the initial data set is isometric to the canonical slice of an extreme Myers-Perry black hole.
\end{abstract}
\maketitle

\section{Introduction}

Based on the standard picture of gravitational collapse for $3+1$ dimensional asymptotically flat spacetimes \cite{choquet2009general}, heuristic physical arguments \cite{dain2012geometric} lead to an inequality relating the total (ADM) mass and angular momentum
\begin{equation} \label{mj4d}
m \geq \sqrt{|\mathcal{J}|},
\end{equation}
if angular momentum is conserved during the evolution. In order to achieve such a property for the angular momentum, axisymmetry is typically imposed along with other conditions on the matter fields. It turns out that it is most natural to treat this inequality at the level of initial data $(M^3,g,k)$, where $g$ is a Riemannian metric on the 3-manifold $M^3$, and $k$ represents the extrinsic curvature of the embedding into spacetime. In this regard, Dain \cite{Dain2008} was the first to rigorously establish \eqref{mj4d} for a general class of vacuum, maximal initial data sets. In this result it was assumed that $M^3\cong\mathbb{R}^{3}\setminus\{0\}$ admits a global Brill (cylindrical) coordinate system $(\rho,z,\phi)$ in which the metric takes the form
\begin{equation}\label{1}
g=e^{2U+2\alpha}(d\rho^{2}+dz^{2})+\rho^{2}e^{2U}(d\phi+A_{\rho}d\rho+A_{z}dz)^{2},
\end{equation}
for some coefficients $U$, $\alpha$, $A_{\rho}$, and $A_{z}$ satisfying appropriate asymptotics.
This particularly simple form of the metric played an important role in the proof. Namely with this, the scalar curvature may be integrated by parts to arrive at a lower bound for the mass, in terms of a (reduced) harmonic energy functional. The second step of the argument then entails showing that the energy functional is minimized by an extreme Kerr harmonic map with the same angular momentum. Later, Chrusciel \cite{chrusciel2008masspositivity} showed that the class of initial data that Dain used was quite general. More precisely, he showed that any simply connected, axisymmetric initial data set with certain asymptotics, admits global Brill coordinates. Further progress was also made with regards to the harmonic map part of the problem. In \cite{schoen2013convexity}, Schoen and Zhou used the convexity properties of harmonic map energies along geodesic deformations in order to simplify the proof, achieve weaker hypotheses on the asymptotics, and to obtain a gap lower bound between the energy of the given data and that of the minimizer. A charged version of \eqref{mj4d} has also been established in \cite{Chrusciel2009,Costa2010,schoen2013convexity}. Corresponding rigidity statements have also been given \cite{Dain2008,khuri2015positive,schoen2013convexity} when these inequalities are saturated, that is the initial data must be the canonical slice of an extreme Kerr or extreme Kerr-Newman black hole. Furthermore, an extension of these inequalities to the case of multiple black holes has been given in \cite{chrusciel2008mass,khuri2015positive}, although here the lower bound for the mass is not known as an explicit quantity.

All of the results mentioned so far involve the maximal assumption, which yields important positivity properties for the scalar curvature. In the nonmaximal case, Zhou \cite{Zhou} has treated \eqref{mj4d} for vacuum initial data with small $\tr_{g}k$, and Cha and the second author have reduced the general case, for both the original inequality \cite{cha2014deformations} and its charged version \cite{cha2015deformations} to solving a canonical system of elliptic equations.

As alluded to above, a closely related topic to the proof of the mass-angular momentum and mass-angular momentum-charge inequality is the uniqueness and existence of stationary, axisymmetric black hole solutions to the vacuum and electrovacuum Einstein equations.  This problem is equivalent to showing uniqueness and existence of certain singular harmonic maps from $\mathbb{R}^3$ into 2-dimensional hyperbolic space $\mathbb{H}^{2}$ and complex hyperbolic space $\mathbb{H}^{2}_{\mathbb{C}}$, respectively.
In particular, in the single black hole case, it is known in this setting \cite{chrusciel2012stationary} that the Kerr(-Newman) family of black hole solutions exhausts all possibilities.  The extreme members of this family then provide the minimizers for the mass lower bound.

The purpose of the present article is to establish a mass-angular momentum inequality in five dimensions. The investigation of higher dimensional black hole solutions has attracted a great deal of interest in recent years \cite{emparan2008black,hollands2012black}, chiefly motivated by string theory and the gauge theory-gravity correspondence.  A central result in this area is the proof by Galloway and Schoen \cite{galloway2006rigidity,Galloway2006} that cross-sections $\mathcal{H}$ of the event horizon, and more generally marginally outer trapped surfaces, must be of positive Yamabe type if the dominant energy condition is satisfied. This implies that $\mathcal{H}$ is diffeomorphic to the sphere $S^3$ or its quotients,  $S^1 \times S^2$, or connected sums thereof.  A second key result due to Hollands, Ishibashi, and Wald \cite{Hollands2007}, and independently by Isenberg and Moncrief \cite{moncrief2008symmetries}, is a rigidity theorem which states that in the analytic setting, a stationary, rotating black hole must admit an additional $U(1)$ isometry.  Explicit vacuum solutions corresponding to $\mathcal{H}\cong S^3$ and $\mathcal{H} \cong S^1 \times S^2$ are known, these are respectively the Myers-Perry family of solutions \cite{Myers1986} and the `black ring' solution of Emparan and Reall \cite{emparan2002rotating} (see also \cite{Pomeransky2006}). More recently, the first example of a black hole with real projective space topology $\mathcal{H} \cong \mathbb{RP}^{3}$ has been found as a solution to supergravity \cite{kunduri2014supersymmetric}.  All these solutions admit $U(1)^2$ isometries.

In order to establish geometric inequalities involving angular momentum for black holes in 5-dimensions, it is natural to consider initial data admitting a $U(1)^2$ action by isometries. As each such inequality is expected to be associated with a model spacetime, which saturates the inequality, it is useful here to recall the basic uniqueness theorem for 5-dimensional stationary vacuum black hole solutions, in order to determine which solutions may serve as models. The first important fact to note is that in five dimensions, black holes are not determined by their mass and charges alone. This is exhibited dramatically by the existence of the black ring solution which can possess the same mass and angular momenta as a Myers-Perry black hole, but with a horizon of different topology. Furthermore, Hollands and Yzadjiev \cite{Hollands2008} have shown that after fixing mass and angular momenta, nondegenerate stationary vacuum black holes with $U(1)^2$ isometries are uniquely determined by a set of invariants which characterize the fixed points of the $U(1)^2$ action and the surfaces on which the timelike Killing field is null; this is referred to as the `orbit space' data. This data encodes, in particular, the topology of the horizon and the second homology group of the domain of outer communication.  An analogous result holds for extreme (degenerate) black holes \cite{Figueras2010}. Interestingly, these results do not address the question of existence of black hole solutions for a given orbit space. However, they indicate that the unique solution (if it exists) associated with each orbit space has the potential to serve as a model black hole for a geometric inequality.

In general, the orbit space is a 2-dimensional manifold with a boundary consisting of 1-dimensional segments and corners. On such segments and corners, respectively one and two linear combinations of the Killing fields generating the $U(1)^2$ isometries have fixed points \cite{alaee2014thesis,alaee2014mass,Hollands2008}. In the present work we will restrict attention to initial data which have the same orbit space structure as that of the Myers-Perry black holes. Here the orbit space may be identified with a half plane minus the origin, in which the boundary consists of two infinitely long rays, each of which serves as the fixed point set for one of the two rotational Killing fields.

Consider an initial data set $(M^4, g, k)$ for the 5-dimensional Einstein equations. Again this consists of a 4-manifold $M^4$, Riemannian metric $g$, and symmetric 2-tensor $k$ representing
extrinsic curvature. The energy and momentum density of the matter fields are given by
\begin{equation}
16\pi\mu = R+(\tr_{g}k)^{2}-|k|_{g}^{2},\text{ }\text{ }\text{ }\text{ }\text{ }\text{ }\text{ }\text{ }
8\pi J = \operatorname{div}_{g}(k-(\tr_{g}k)g),
\end{equation}
where $R$ is the scalar curvature of $g$.
It will be assumed throughout that the data are bi-axially symmetric. This means that the group of isometries of the
Riemannian manifold $(M^4,g)$ has a subgroup
isomorphic to $U(1)^2$ with no discrete isotropy subgroups, and that all quantities defining the initial data are invariant under the $U(1)^2$ action. Thus if
$\eta_{(l)}$, $l=1,2$ are the two Killing field generators associated with this symmetry, then
\begin{equation}
\mathfrak{L}_{\eta_{(l)}}g=\mathfrak{L}_{\eta_{(l)}}k
=\mathfrak{L}_{\eta_{(l)}}\mu=\mathfrak{L}_{\eta_{(l)}}J=0,
\end{equation}
where $\mathfrak{L}_{\eta_{(l)}}$ denotes Lie differentiation. We will also postulate that
$M^4$ has two ends, with one designated end being asymptotically flat, and the other being either asymptotically
flat or asymptotically cylindrical. Recall that a domain $M^{4}_{\text{end}}\subset M^4$ is an
asymptotically flat end if it is diffeomorphic to $\mathbb{R}^{4}\setminus\text{Ball}$, and in the coordinates given by the asymptotic diffeomorphism the following fall-off conditions hold
\begin{equation}\label{3}
g_{ab}=\delta_{ab}+O_{1}(r^{-1-\kappa}),\text{ }\text{ }\text{ }\text{ }\text{ }
k_{ab}=O(r^{-2-\kappa}),\text{
}\text{ }\text{ }\text{ }
\text{ }\mu\in L^{1}(M^{4}_{\text{end}}),\text{ }\text{ }\text{ }\text{ }\text{ }
J_{i}\in L^{1}(M^{4}_{\text{end}}),
\end{equation}
for some $\kappa>0$.
These asymptotics guarantee that the ADM energy and linear momentum are well-defined, with the energy given by the following limit
\begin{equation}\label{6}
m=\frac{1}{16\pi}\int_{S_{\infty}}(g_{ab,a}-g_{aa,b})\nu^{b},
\end{equation}
where $S_{\infty}$ indicates the limit as $r\rightarrow\infty$ of integrals over coordinate spheres $S_{r}$, with unit outer normal $\nu$. Although the asymptotics \eqref{3} are not strong enough to ensure that the linear momentum vanishes, and so the mass does not coincide with the energy, we will throughout this paper refer to the quantity \eqref{6} as the mass in order to reserve the use of the term `energy' in reference to harmonic maps. We note that the weaker hypothesis $g_{ab},p_{ab}\in L^{2}(M^{4}_{\text{end}})$ may be used in place of the explicit asymptotics involving $\kappa$, where $p=k-(\tr_{g} k)g$ is the momentum tensor, in order to achieve well-defined ADM energy-momentum. Moreover, it is likely that the results of this paper hold under these weaker conditions, but we will not pursue such questions here. Now consider the ADM angular momenta
\begin{equation}\label{7}
\mathcal{J}_{l}=\frac{1}{8\pi}\int_{S_{\infty}}(k_{ab}-(\tr_{g} k)g_{ab})\nu^{a}\eta_{(l)}^{b},
\text{ }\text{ }\text{ }\text{ }\text{ }l=1,2.
\end{equation}
A priori this may not yield a finite well-defined quantity solely under the asymptotics \eqref{3}, since the Killing fields grow like $r^{2}$. However, under the additional assumption that $J(\eta_{(l)})\in L^{1}(M^{4}_{\text{end}})$, $l=1,2$ we have that \eqref{7} is finite. This may easily be seen by integrating the following expression over $M^{4}_{\text{end}}$,
\begin{equation}
\operatorname{div}_{g} p(\eta_{(l)})=(\operatorname{div}_{g} p)(\eta_{(l)})
+\frac{1}{2}\langle p,\mathfrak{L}_{\eta_{(l)}}g\rangle=8\pi J(\eta_{(l)}).
\end{equation}
Our main result is as follows.

\begin{theorem}\label{TheoremI}
Let $(M^4,g,k)$ be a smooth, complete, bi-axially symmetric, maximal initial data set for the 5-dimensional Einstein equations satisfying $\mu\geq 0$
and $J(\eta_{(l)})=0$, $l=1,2$ and with two ends, one designated asymptotically flat and the other either asymptotically flat or
asymptotically cylindrical.
If $M^{4}$ is diffeomorphic to $\mathbb{R}^{4}\setminus\{0\}$ and admits a global system of generalized Brill coordinates then
\begin{equation}\label{23}
m^3\geq \frac{27\pi}{32}\left(|\mathcal{J}_1|+|\mathcal{J}_2|\right)^2.
\end{equation}
Moreover if $\mathcal{J}_{i}\neq 0$, $i=1,2$, then equality holds if and only if $(M^4,g,k)$ is isometric to the canonical slice of an extreme Myers-Perry spacetime.
\end{theorem}

This theorem may be considered as a direct generalization of Dain's result \cite{Dain2008} to higher dimensions, as both assume the existence of a global Brill coordinate system. It also generalizes the local versions of inequality \eqref{23} established in \cite{alaee2015proof, Alaeeremarks2015}, for data which are sufficiently close to extreme Myers-Perry. Moreover, this result may be interpreted as giving a variational characterization of the extreme Myers-Perry initial data, as the mass minimizers among all data with fixed angular momentum. Note that the horizon geometries of 5-dimensional extreme vacuum black holes also arise as minimizers in the context of the area-angular momenta inequalities proved in \cite{Hollands2012}; such minimizers have been completely classified \cite{Kunduri2009}.

The assumption of nonvanishing angular momenta is included since otherwise there is no extreme Myers-Perry black hole to serve as a model; the extreme Myers-Perry solutions with one or more vanishing angular momenta do not contain a black hole. In particular, the inequality when $\mathcal{J}_{i}=0$, $i=1,2$ reduces to the positive mass theorem, and due to the topology of the initial data the case of equality cannot be achieved. Let us now make a few remarks concerning the other hypotheses. From the preceding discussion the motivation for most of the hypotheses should be clear, except perhaps those associated with the momentum density and Brill coordinates. The assumption $J(\eta_{(l)})=0$, $l=1,2$ is, as mentioned above, used to obtain well-defined total angular momenta, but will also be used for the important purpose of guaranteeing the existence of twist potentials, which encode the relevant information concerning angular momentum and help reduce the proof to a harmonic map problem. The existence of a generalized Brill coordinate system ensures that there is a global system of (cylindrical) coordinates such that the metric takes a simple form analogous to \eqref{1} in the $3+1$ dimensional case; a precise description will be given in the next section. Although this appears to be a restrictive assumption, let us recall that in the $3+1$ setting, Chrusciel \cite{chrusciel2008masspositivity} (see also \cite{KhuriSokolowsky}) has shown under general conditions that simply connected axisymmetric initial data admit global Brill coordinates. Similarly, we conjecture that under  appropriate asymptotics, a simply connected bi-axisymmetric initial data set with trivial second homology group admits a desired set of generalized Brill coordinates.

A natural question to ask is whether the current theorem admits generalizations to dimensions higher than 5. It turns out that this is not possible if we require the data to be asymptotically flat with a 2-dimensional orbit space. To see this, suppose that a spacetime has dimension $n$ with a $U(1)^{n-3}$ symmetry. Asymptotic flatness implies that the initial data will have $SO(n-1)$ as the compact part of the asymptotic symmetry group, however this special orthogonal group admits at most $(n-1)/2$ mutually commuting generators. Thus, the only dimensions for which $U(1)^{n-3}\subset SO(n-1)$ are $n=4,5$.

This paper is organized as follows. In Section 2 we give a detailed description of generalized Brill coordinates. In Section 3 we derive a lower bound for the mass in terms of a functional that will be related to the harmonic energy of a map from $\mathbb{R}^{3}\rightarrow SL(3,\mathbb{R})/SO(3)$. Section 4 is then dedicated to proving that the extreme Myers-Perry harmonic map achieves the absolute minimum of this functional, and at the end of this section we then prove Theorem \ref{TheoremI}. A discussion of future directions and generalizations of the results presented here is given in Section 5. Finally an appendix is included to record, among other things, important properties of the Myers-Perry black holes.

\section{Generalized Brill Coordinates}

In this section we seek a certain type of cylindrical coordinate system for the initial data, which are isothermal for the metric induced on the orbit space. These generalized Brill coordinates are related to the well-known Weyl coordinates familiar from the Ernst reduction of the stationary vacuum Einstein equations. The primary difference between Brill and Weyl coordinates is that the former applies to the region inside and outside of a black hole, while the latter only covers the outer region. However, in the case of extreme black holes the two types of coordinates coincide. The following definition was initially given in \cite{alaee2014thesis,Alaeeremarks2015} in a more general context, whereas here it is refined for the particular problem at hand.

\begin{definition}\label{GBdata}
An initial data set $(M^{4},g,k)$ with a $U(1)^{2}$ symmetry, and $M^{4}\cong \mathbb{R}^{4}\setminus \{0\}$, is said to admit a system of generalized Brill coordinates $(\rho,z,\phi^{1},\phi^{2})$ if globally the metric takes the form
\begin{equation}\label{GBmetric}
g=\frac{e^{2U+2\alpha}}{2\sqrt{\rho^2+z^2}}\left(d\rho^2+d z^2\right)+e^{2U}\lambda_{ij}\left(d\phi^i+A^i_l d y^l\right)\left(d\phi^j+A^j_l d y^l\right),
\end{equation}
for some functions $U$, $\alpha$, $A_{l}^{i}$, and a symmetric positive definite matrix $\lambda=(\lambda_{ij})$ with $\det\lambda=\rho^{2}$, $i,j,l=1,2$, $(y^1,y^2)=(\rho,z)$, all independent of $(\phi^{1},\phi^{2})$ and satisfying the asymptotics \eqref{F1}-\eqref{F12}. Moreover, the coordinates should take values in the following ranges $\rho\in[0,\infty)$, $z\in \mathbb{R}$, and $\phi^i\in [0,2\pi]$, $i=1,2$.
\end{definition}

The similarity of the metric structure \eqref{GBmetric} above with that of traditional Brill coordinates in 3-dimensions \eqref{1} is evident, except perhaps for the presence of $\sqrt{\rho^{2}+z^{2}}$. This term is included so that like in the 3-dimensional case, \eqref{GBmetric} reduces to the flat metric when $U=\alpha=A_{l}^{i}=0$, if $\lambda=\sigma:=r^{2}\operatorname{diag}(\sin^{2}\theta,\cos^{2}\theta)$, where the appropriate polar coordinates are given by
\begin{equation}\label{polarcoord}
\rho=\frac{1}{2}r^{2}\sin(2\theta),\text{ }\text{ }\text{ }\text{ }\text{ }\text{ }\text{ }
z=\frac{1}{2}r^{2}\cos(2\theta),\text{ }\text{ }\text{ }\text{ }\text{ }\text{ }\text{ }
r^{2}=2\sqrt{\rho^{2}+z^{2}},
\end{equation}
with $r\in [0,\infty)$, $\theta\in[0,\pi/2]$. Note that the coordinates $(\theta,\phi^{1},\phi^{2})$ are Hopf coordinates for the 3-sphere, which are naturally associated with the Hopf fibration. In particular, the flat metric is given in these two coordinate systems by
\begin{equation}\label{flatmet}
\delta_4= \frac{d\rho^2 + d z^2}{2\sqrt{\rho^2 + z^2} } + \sigma_{ij}d\phi^{i} d\phi^{j}=
d r^2 + r^2d\theta^2 +r^{2}\left(\sin^{2}\theta (d\phi^1)^{2}+\cos^{2}\theta (d\phi^2)^{2}\right).
\end{equation}
We also note that without loss of generality the generators of the $U(1)^{2}$ symmetry may be chosen such that $\eta_{(l)} = \partial_{\phi^l}$, $l=1,2$.

Let us now record the appropriate asymptotics in three different regions, namely at infinity, the origin, and near the axis. The particular decay rates are motivated in general by the indicated asymptotically flat and asymptotically cylindrical geometries, and by the desire for certain coefficients, including $\lambda_{ij}$ and $A^{i}_{l}$, to not yield a direct contribution to the ADM mass. In what follows $a,b$ are functions of $\theta$, $\kappa>0$ is as in the previous section, and $\tilde{\sigma}=\tilde{\sigma}_{ij}d\phi^{i}d\phi^{j}$ is a Riemannian metric on the torus $T^2$ depending only on $\theta$. We begin with the designated asymptotically flat end characterized by $r\rightarrow\infty$, and require
\begin{equation}\label{F1}
U=O_{1}(r^{-1-\kappa}),\text{ }\text{ }\text{ }\text{ }\text{ }
\alpha=O_1(r^{-1-\kappa}),\text{ }\text{ }\text{ }\text{ }\text{ }
A_{\rho}^i=\rho O_{1}(r^{-5-\kappa}),
\text{ }\text{ }\text{ }\text{ }\text{ }
A_z^i=O_{1}(r^{-3-\kappa}),
\end{equation}
\begin{equation}\label{F2}
\lambda_{ii}=\left(1+(-1)^{i}ar^{-1-\kappa}+O_{1}(r^{-2-\kappa})\right)\sigma_{ii},\text{ }\text{ }\text{ }\text{ }\text{ }\text{ }\text{ }\lambda_{12}=\rho^2 O_1(r^{-5-\kappa}),
\text{ }\text{ }\text{ }\text{ }\text{ }\text{ }\text{ }|k|_{g}=O(r^{-2-\kappa}).
\end{equation}
Next consider the asymptotics as $r\rightarrow 0$, where there are two types to account for.  Namely, in the asymptotically flat case
\begin{equation}\label{F4}
U=-2\log r+O_{1}(1),\text{ }\text{ }\text{ }\text{ }\text{ }
\alpha=O_1(r^{1+\kappa}),\text{ }\text{ }\text{ }\text{ }\text{ }
A_{\rho}^i=\rho O_{1}(r^{1+\kappa}),
\text{ }\text{ }\text{ }\text{ }\text{ }
A_z^i=O_{1}(r^{3+\kappa}),
\end{equation}
\begin{equation}\label{F5}
\lambda_{ii}=\left(1+(-1)^{i}br^{1+\kappa}+O_{1}(r^{2+\kappa})\right)\sigma_{ii},\text{ }\text{ }\text{ }\text{ }\text{ }\text{ }\text{ }\lambda_{12}=\rho^2 O_1(r^{-\frac{1}{2}+\kappa}),\text{ }\text{ }\text{ }\text{ }\text{ }\text{ }\text{ }
|k|_{g}=O(r^{2+\kappa}),
\end{equation}
and in the asymptotically cylindrical case
\begin{equation}\label{F7}
U=-\log r+O_{1}(1),\text{ }\text{ }\text{ }\text{ }\text{ }\text{ }
\alpha=O_{1}(1),\text{ }\text{ }\text{ }\text{ }\text{ }\text{ }A_{\rho}^i=\rho O_{1}(r^{1+\kappa}),
\text{ }\text{ }\text{ }\text{ }\text{ }\text{ }
A_z^i=O_{1}(r^{3+\kappa}),
\end{equation}
\begin{equation}\label{F8}
\lambda_{ij}=r^{2}\tilde{\sigma}_{ij}+O_{1}(r^{2+\kappa}),\text{ }\text{ }\text{ }\text{ }\text{ }\text{ }\text{ }\text{ }|k|_{g}=O(r^{2+\kappa}).
\end{equation}
Lastly, let $\Gamma=\Gamma_{+}\cup\Gamma_{-}$ denote the two axes $\Gamma_{\pm}=\{\rho=0, \pm z>0\}$, then the asymptotics as $\rho\rightarrow 0$ are given by
\begin{equation}\label{F10}
U=O_{1}(1),\text{ }\text{ }\text{ }\text{ }\text{ }\text{ }
\alpha=O_1(1),\text{ }\text{ }\text{ }\text{ }\text{ }\text{ }
A_{\rho}^i= O_{1}(\rho),
\text{ }\text{ }\text{ }\text{ }\text{ }\text{ }
A_z^i=O_{1}(1),\text{ }\text{ }\text{ }\text{ }\text{ }\text{ }
|k|_{g}=O(1),
\end{equation}
\begin{equation}\label{F12}
\lambda_{11},\lambda_{12}=O(\rho^{2}),\text{ }\text{ }\text{ }\text{ }\text{ }
\lambda_{22}=O(1)
\text{ }\text{ }\text{ on }\text{ }\text{ }\Gamma_{+},\text{ }\text{ }\text{ }\text{ }\text{ }
\lambda_{22},\lambda_{12}=O(\rho^{2}),\text{ }\text{ }\text{ }\text{ }\text{ }
\lambda_{11}=O(1)
\text{ }\text{ }\text{ on }\text{ }\text{ }\Gamma_{-}.
\end{equation}

It should be pointed out that regularity of the geometry along the axis implies a compatibility condition between $\alpha$ and $\lambda$. To see this, let $\vartheta\in(-\infty,2\pi)$ be the
cone angle deficiency coming from the metric $g$ at the axes of rotation, that is
\begin{equation}
\frac{2\pi}{2\pi-\vartheta}=\lim_{\rho\rightarrow 0}\frac{2\pi\cdot\mathrm{Radius}}{\mathrm{Circumference}}
=\lim_{\rho\rightarrow 0}\frac{{\displaystyle \int_{0}^{\rho}}\sqrt{\frac{e^{2U+2\alpha}}{2\sqrt{\rho^{2}+z^{2}}}
+e^{2U}\lambda_{ij}A_{\rho}^{i}A_{\rho}^{j}}d\rho}
{\sqrt{e^{2U}\lambda_{ii}}}=\frac{e^{\alpha(0,z)}}{\sqrt{2|z|}}\lim_{\rho\rightarrow 0}
\frac{\rho}{\sqrt{\lambda_{ii}}}
\end{equation}
where $i=1,2$ corresponds to $\Gamma_{+},\Gamma_{-}$, respectively.
The cone angle deficiency should vanish $\vartheta=0$, since $(M^4,g)$ is smooth across the axis, and thus
\begin{equation}\label{cone}
\alpha(0,z)=\frac{1}{2}\log\left(|z|\partial_{\rho}^{2}\lambda_{ii}(0,z)\right)=:\alpha_{\pm}(z)\text{ }\text{ }\text{ }\text{ on }\text{ }\text{ }\text{ }\Gamma_{\pm}.
\end{equation}

To close this section, we confirm here that the asymptotically flat asymptotics \eqref{F1}, \eqref{F2} and \eqref{F4}, \eqref{F5} used to define Brill coordinates are consistent with those given in \eqref{3}. First observe that the fall-off imposed on $k$ in \eqref{F2} trivially implies that in \eqref{3}. Consider now the cartesian coordinates
\begin{equation}
x^1=r\cos\theta\cos\phi^1,\quad x^2=r\cos\theta\sin\phi^1,\quad
x^3=r\sin\theta\cos\phi^2,\quad x^4=r\sin\theta\sin\phi^2.
\end{equation}
Upon expressing the metric in these coordinates it follows that
\begin{align}
\begin{split}
g=&e^{2U+2\alpha}(dr^2+r^{2}d\theta^2)+e^{2U}\lambda_{ij}\left(d\phi^i+A^i_l d y^l\right)\left(d\phi^j+A^j_l d y^l\right)\\
=&\delta+\underbrace{(e^{2U+2\alpha}-1)}_{O_{1}(r^{-1-\kappa})}\underbrace{(dr^2+r^{2}d\theta^2)}_{O_{1}(1)}
+\underbrace{(e^{2U}\lambda_{ij}-\sigma_{ij})}_{O_{1}(r^{-1-\kappa})}
\underbrace{d\phi^{i}d\phi^{j}}_{O_{1}(r^{-2})}\\
&+\underbrace{e^{2U}\lambda_{ij}}_{O_{1}(1)}
\underbrace{A_{l}^{j}}_{O_{1}(r^{-3-\kappa})}
\underbrace{d\phi^{i}dy^{l}}_{O_{1}(1)}+\underbrace{e^{2U}\lambda_{ij}}_{O_{1}(r^{2})}
\underbrace{A_{l}^{i}dy^{l}}_{O_{1}(r^{-2-\kappa})}
(\underbrace{d\phi^{j}}_{O_{1}(r^{-1})}+\underbrace{A_{l}^{j}dy^{l}}_{O_{1}(r^{-2-\kappa})})\\
=&\delta+O_{1}(r^{-1-\kappa}),
\end{split}
\end{align}
where we have used $d \rho = O(r)$ and $dz = O(r)$. Similar computations yield the same result for the asymptotics \eqref{F4}, \eqref{F5}.

\section{The Mass Functional}

One of the advantages of Brill data is that it provides a particularly simple expression for the scalar curvature. Namely, as shown in \cite{Alaeeremarks2015} we have
\begin{equation}\label{SCALAR}
e^{2U+2\alpha-2\log r}R=-6\Delta U-2\Delta_{\rho,z}\alpha-6|\nabla U|^2
+\frac{\det\nabla\lambda}{2\rho^{2}}\\
-\frac{1}{4}e^{-2\alpha+2\log r}\lambda_{ij}(A_{\rho,z}^{i}-A_{z,\rho}^{i})
(A_{\rho,z}^{j}-A_{z,\rho}^{j}),
\end{equation}
where $\Delta$ and the norm $|\cdot|$ are with respect to the following flat metric
\begin{equation}
\delta_{3}=r^{2}\left(d r^{2}+r^{2}d\theta^{2}\right)+\frac{r^{4}\sin^{2}(2\theta)}{4} d\phi^{2}
=d\rho^{2}+d z^{2}+\rho^{2} d\phi^{2}
\end{equation}
on an auxiliary $\mathbb{R}^{3}$ in which all quantities are independent of the new variable $\phi\in[0,2\pi]$, and $\Delta_{\rho,z}$ is with respect to the flat metric
$\delta_{2}=d\rho^{2}+dz^{2}$ on the orbit space. Moreover, the notation used for the last term on the first line is shorthand for
\begin{equation}
\det\nabla\lambda
=\det
\begin{pmatrix}
\nabla\lambda_{11}&\nabla\lambda_{12}\\
\nabla\lambda_{12}&\nabla\lambda_{22}
\end{pmatrix}
=\delta_{3}(\nabla\lambda_{11},\nabla\lambda_{22})-|\nabla \lambda_{12}|^2.
\end{equation}
From \eqref{SCALAR} one may integrate by parts to obtain a closed form expression \cite{Alaeeremarks2015} for the mass
\begin{align}\label{massf}
\begin{split}
m=&\frac{1}{8}\int_{\mathbb{R}^{3}}\left(e^{2U+2\alpha-2\log r}R
+6|\nabla U|^{2}-\frac{\det\nabla\lambda}{2\rho^{2}}\right)dx\\
&+\frac{1}{32}\int_{\mathbb{R}^{3}}
e^{-2\alpha+2\log r}\lambda_{ij}(A_{\rho,z}^{i}-A_{z,\rho}^{i})
(A_{\rho,z}^{j}-A_{z,\rho}^{j})dx
+\frac{\pi}{2}\sum_{\varsigma=\pm}\int_{\Gamma_{\varsigma}}\alpha_{\varsigma}dz,
\end{split}
\end{align}
where the volume form $dx$ is again with respect to $\delta_{3}$.

The next goal is to relate the right-hand side of \eqref{massf} to a reduced form of a harmonic energy. In order to accomplish this, the scalar curvature will be replaced by an expression involving potentials for the angular momentum. Consider the 1-form
\begin{equation}\label{123456}
\mathcal{P}_{(l)}=2\star \left(p(\eta_{(l)})\wedge\eta_{(1)}\wedge\eta_{(2)}\right)
=2\epsilon_{abcd}p^{b}_{s}\eta_{(l)}^{s}\eta_{(1)}^{c}\eta_{(2)}^{d}dx^{a}
\end{equation}
on $M^{4}$, where $\epsilon_{abcd}$ is the volume form for $g$, $\star$ is the Hodge star, and $p$ is the momentum tensor. A computation, utilizing the momentum constraint and the fact that $\eta_{(l)}$ is a Killing field, then shows that
\begin{equation}
d\mathcal{P}_{(l)}=-8\pi J(\eta_{(l)})\epsilon_{abcd}\eta_{(1)}^{c}\eta_{(2)}^{d}dx^{a}\wedge dx^{b}.
\end{equation}
Thus, under the the assumptions that $J(\eta_{(l)})=0$, $l=1,2$ and $M^{4}$ is simply connected, twist potentials exist such that
\begin{equation}
d\zeta^{l}=\mathcal{P}_{(l)}, \text{ }\text{ }\text{ }\text{ }\text{ }\text{ }l=1,2.
\end{equation}
It is then clear from \eqref{123456} that $\partial_{z}\zeta^{l}|_{\Gamma}=0$, so that $\zeta^{l}$ is constant on each axis $\Gamma_{\pm}$. These constants in turn determine the ADM angular momenta
\begin{align}\label{a6}
\begin{split}
\mathcal{J}_{l}&=\frac{1}{8\pi}\int_{S_{\infty}}p(\eta_{(l)},\nu)\\
&=\lim_{r\rightarrow0}\frac{1}{8\pi}\int_{\partial B(r)}k(\partial_{\phi^{l}},\nu)dV\\
&=\lim_{r\rightarrow0}\frac{1}{16\pi}\int_{\partial B(1)}k(\partial_{\phi^{l}},\nu)e^{3U+\alpha}r^{3}\sin(2\theta) d\theta d\phi^{1}d\phi^{2}\\
&=\lim_{r\rightarrow0}\frac{1}{16\pi}\int_{\partial B(1)}\partial_{\theta}\zeta^{l} d\theta d\phi^{1}d\phi^{2}\\
&=\frac{\pi}{4}(\zeta^{l}|_{\Gamma_{-}}-\zeta^{l}|_{\Gamma_{+}}),
\end{split}
\end{align}
where $B(r)$ is the coordinate ball of radius $r$ centered at the origin. Furthermore, consider the frame
\begin{equation}
e_1= e^{-U-\alpha+\log r}\left(\partial_\rho - A^{i}_{\rho} \partial_{\phi^i}\right), \text{ }\text{ }\text{ }e_2 = e^{-U-\alpha+\log r}\left(\partial_z - A^{i}_{z} \partial_{\phi^i}\right),\text{ }\text{ }\text{ }
 e_{i+2} = e^{-U} \partial_{\phi^i},\text{ }\text{ }i=1,2,
\end{equation}
with dual co-frame
\begin{equation}
\theta^1= e^{U+\alpha-\log r}d\rho,\text{ } \text{ } \text{ } \text{ } \text{ } \text{ }
\theta^2= e^{U+\alpha-\log r}dz,\text{ }\text{ }\text{ }\text{ }\text{ }\text{ }
\theta^{i+2}=e^{U}\left( d\phi^i+A^i_l d y^l\right),\text{ }\text{ }\text{ }\text{ }i=1,2,
\end{equation}
so that the metric may be written as
\begin{equation}
g=(\delta_{2})_{ln}\theta^{l}\theta^{n}+\lambda_{ij}\theta^{i+2}\theta^{j+2},
\end{equation}
and
\begin{equation}\label{equation}
k(e_1,e_{i+2}) = -\frac{e^{-4U-\alpha+\log r}}{2\rho} \partial_z \zeta^{i},\text{ }\text{ }\text{ }\text{ }\text{ }\text{ }\text{ }
k(e_2,e_{i+2})= \frac{e^{-4U-\alpha+\log r}}{2\rho} \partial_\rho \zeta^{i}.
\end{equation}
Therefore, in light of the maximal condition $\tr_{g}k=0$ we have
\begin{align}\label{asdf}
\begin{split}
R=&16\pi\mu+|k|_{g}^{2}\\
=&16\pi\mu+\frac{e^{-8U-2\alpha+2\log r}}{2\rho^{2}}\nabla\zeta^{t}\lambda^{-1}\nabla\zeta\\
&+k(e_{1},e_{1})^{2}+2k(e_{1},e_{2})^{2}+k(e_{2},e_{2})^{2}
+\lambda^{ij}\lambda^{ln}k(e_{i},e_{l})k(e_{j},e_{n}),
\end{split}
\end{align}
where
\begin{equation}
\nabla \zeta^t\lambda^{-1}\nabla \zeta
=\begin{pmatrix}
\nabla \zeta^1&
\nabla \zeta^2
\end{pmatrix}
\begin{pmatrix}
\lambda^{11}&\lambda^{12}\\
\lambda^{12}&\lambda^{22}
\end{pmatrix}
\begin{pmatrix}
\nabla \zeta^1\\
\nabla \zeta^2
\end{pmatrix}
=\sum_{i,j=1,2}\lambda^{ij}\delta_{3}(\nabla\zeta^{i},\nabla\zeta^{j}).
\end{equation}
It follows that by combining \eqref{massf} and \eqref{asdf}
\begin{align}\label{lkjh}
\begin{split}
m=&\mathcal{M}(U,\lambda,\zeta)+\frac{1}{8}\int_{\mathbb{R}^{3}}\left(16\pi e^{2U+2\alpha-2\log r}\mu
+\frac{1}{4}e^{-2\alpha+2\log r}\lambda_{ij}(A_{\rho,z}^{i}-A_{z,\rho}^{i})
(A_{\rho,z}^{j}-A_{z,\rho}^{j})\right)dx\\
&+\frac{1}{8}\int_{\mathbb{R}^{3}}e^{2U+2\alpha-2\log r}\left(k(e_{1},e_{1})^{2}+2k(e_{1},e_{2})^{2}+k(e_{2},e_{2})^{2}
+\lambda^{ij}\lambda^{ln}k(e_{i},e_{l})k(e_{j},e_{n})\right)dx,
\end{split}
\end{align}
where
\begin{equation}\label{MASSFUN}
\mathcal{M}(U,\lambda,\zeta)
=\frac{1}{8}\int_{\mathbb{R}^{3}}
\left(6|\nabla U|^{2}-\frac{\det\nabla\lambda}{2\rho^{2}}+\frac{e^{-6U}}{2\rho^{2}}\nabla\zeta^{t}
\lambda^{-1}\nabla\zeta\right)dx
+\frac{\pi}{2}\sum_{\varsigma=\pm}\int_{\Gamma_{\varsigma}}\alpha_{\varsigma}dz.
\end{equation}

The mass functional $\mathcal{M}$ is to be related to a reduced harmonic energy. However, it is not even immediately apparent from the expression in \eqref{MASSFUN} that this quantity is nonnegative in general. It turns out that this may be resolved with an appropriate transformation or change of variables $(\lambda_{11},\lambda_{22},\lambda_{12})\rightarrow (V,W)$; note that since $\det\lambda=\rho^{2}$ there are only two independent functions contained in $\lambda$. Define the new variables by
\begin{equation}\label{VW}
V=\frac{1}{2}\log\left(\frac{\lambda_{11}\cos^{2}\theta}{\lambda_{22}\sin^{2}\theta}\right),
\text{ }\text{ }\text{ }\text{ }\text{ }\text{ }W=\sinh^{-1}\left(\frac{\lambda_{12}}{\rho}\right),
\end{equation}
and note the inverse transformation is then given by
\begin{equation}\label{inverse}
\lambda_{11}=\left(\sqrt{\rho^2+z^2}-z\right)e^{V}\cosh W, \text{ }\text{ }\text{ }\text{ }\text{ }
\lambda_{22}=\left(\sqrt{\rho^2+z^2}+z\right)e^{-V}\cosh W,\text{ }\text{ }\text{ }\text{ }\text{ }
\lambda_{12}=\rho\sinh W.
\end{equation}
From the third equation in \eqref{inverse}, and \eqref{F12}, we find that $W=0$ on $\Gamma$. Using this fact together with the first two equations in \eqref{inverse}, and recalling that there are no conical singularities \eqref{cone}, shows that
\begin{equation}\label{V}
V=2\alpha_{+}\text{ }\text{ }\text{ on }\text{ }\text{ }\Gamma_{+},\text{ }\text{ }\text{ }\text{ }\text{ }\text{ } V=-2\alpha_{-}\text{ }\text{ }\text{ on }\text{ }\text{ }\Gamma_{-}.
\end{equation}
Consider now the following harmonic functions on $(\mathbb{R}^{3}\setminus\Gamma,\delta_{3})$ which are naturally associated with the above transformation:
\begin{equation}\label{harmonicfunction}
h_{1}=\frac{1}{2}\log\rho,\text{ }\text{ }\text{ }\text{ }\text{ }\text{ }\text{ }\text{ }\text{ }h_{2}=\frac{1}{2}\log\left(\frac{\sqrt{\rho^{2}+z^{2}}-z}{
\sqrt{\rho^2+z^{2}}+z}\right).
\end{equation}
In particular a computation yields
\begin{equation}
-\frac{\det\nabla\lambda}{\rho^2}=|\nabla V|^2+|\nabla W|^2+\sinh^{2}W\abs{\nabla\left(V+h_2\right)}^2+2\delta_{3}(\nabla h_2,\nabla V),
\end{equation}
and the last term may be integrated away to the boundary
\begin{align}
\begin{split}
\frac{1}{8}\int_{\mathbb{R}^{3}}\delta_{3}(\nabla h_{2},\nabla V)dx
=&-\lim_{\varepsilon\rightarrow 0}\frac{1}{8}\int_{\rho=\varepsilon}V\partial_{\rho}h_{2}\\
=&\frac{\pi}{4}\left(\int_{\Gamma_{-}}Vdz-\int_{\Gamma_{+}}Vdz\right)\\
=&-\frac{\pi}{2}\sum_{\varsigma=\pm}\int_{\Gamma_{\varsigma}}\alpha_{\varsigma}dz.
\end{split}
\end{align}
Notice that this boundary term cancels the one in \eqref{MASSFUN}, and so
it follows that
\begin{align}\label{MASSFUNCTION}
\begin{split}
\mathcal{M}(U,V,W,\zeta^{1},\zeta^{2})=&\frac{1}{16}\int_{\mathbb{R}^{3}}12|\nabla U|^{2}+|\nabla V|^{2}+|\nabla W|^{2}
+\sinh^{2}W|\nabla (V+h_{2})|^{2}
dx\\
&
+\frac{1}{16}\int_{\mathbb{R}^{3}}
e^{-6h_{1}-6U+h_{2}+V}\cosh W
\left|e^{-h_{2}-V}\tanh W\nabla \zeta^{1}-\nabla \zeta^{2}\right|^{2}dx\\
&+\frac{1}{16}\int_{\mathbb{R}^{3}}
\frac{e^{-6h_{1}-6U-h_{2}-V}}{\cosh W}|\nabla \zeta^{1}|^{2}dx.
\end{split}
\end{align}
This version of the mass functional is clearly nonnegative, and together with \eqref{lkjh} it establishes the positive mass theorem for generalized Brill initial data. In the next section we will relate this mass functional to a harmonic energy, and establish the mass-angular momentum inequality.

\section{Convexity and the Global Minimizer}
\label{sec4}

Consider the symmetric space $SL(3,\mathbb{R})/SO(3)\cong\mathbb{R}^{5}$ endowed (\cite{helgason1979differential}, \cite{jost2008riemannian}, \cite{maison1979ehlers}) with the nonpositively curved metric
\begin{equation}\label{16}
ds^{2}=12du^{2}+\cosh^{2}w \text{ }\!dv^{2}+dw^{2}
+\frac{e^{-(6u+v)}}{\cosh w}(d\zeta^{1})^{2}
+e^{-6u+v}\cosh w
\left(e^{-v}\tanh w \text{ }\!d\zeta^{1}-d\zeta^{2}\right)^{2}.
\end{equation}
The harmonic energy of a map $\tilde{\Psi}=(u,v,w,\zeta^{1},\zeta^{2}):\mathbb{R}^{3}\rightarrow
SL(3,\mathbb{R})/SO(3)$, on a domain $\Omega\subset\mathbb{R}^{3}$, is then given by
\begin{align}\label{energy1}
\begin{split}
E_{\Omega}(\tilde{\Psi})=&\int_{\Omega}12|\nabla u|^{2}
+\cosh^{2}w|\nabla v|^{2}
+|\nabla w|^{2}
+\frac{e^{-6u-v}}{\cosh w}|\nabla \zeta^{1}|^{2}dx
\\
&+\int_{\Omega}
e^{-6u+v}\cosh w
\left|e^{-v}\tanh w\nabla \zeta^{1}-\nabla \zeta^{2}\right|^{2}dx.
\end{split}
\end{align}
If $\Omega$ has a trivial intersection with the axes of rotation $\Gamma=\{\rho=0\}$, and we write $u=U+h_{1}$, $v=V+h_{2}$, and $w=W$ where $h_{1}$ and $h_{2}$ are the harmonic functions defined in \eqref{harmonicfunction},
then with an integration by parts the reduced energy $\mathcal{I}_{\Omega}$ of the map $\Psi=(U,V,W,\zeta^{1},\zeta^{2})$ may be expressed in terms the harmonic energy of $\tilde{\Psi}$ by
\begin{equation}\label{51}
\mathcal{I}_{\Omega}(\Psi)=E_{\Omega}(\tilde{\Psi})
-12\int_{\partial\Omega}
(h_{1}+2U)\partial_{\nu}h_{1}-\int_{\partial\Omega}(h_{2}+2V)\partial_{\nu}h_{2},
\end{equation}
where $\nu$ denotes the unit outer normal to the boundary $\partial\Omega$ and
\begin{align}
\begin{split}
\mathcal{I}_{\Omega}(\Psi)=&\int_{\Omega}12|\nabla U|^{2}+|\nabla V|^{2}+|\nabla W|^{2}
+\sinh^{2}W|\nabla (V+h_{2})|^{2}
+\frac{e^{-6h_{1}-6U-h_{2}-V}}{\cosh W}|\nabla \zeta^{1}|^{2}dx
\\
&+\int_{\Omega}
e^{-6h_{1}-6U+h_{2}+V}\cosh W
\left|e^{-h_{2}-V}\tanh W\nabla \zeta^{1}-\nabla \zeta^{2}\right|^{2}dx.
\end{split}
\end{align}
Observe that $\mathcal{I}=\mathcal{I}_{\mathbb{R}^{3}}=16\mathcal{M}$
where $\mathcal{M}$ is the mass functional \eqref{MASSFUNCTION}. The reduced energy $\mathcal{I}$ may be considered a regularization of $E$ since the infinite terms $\int|\nabla h_{1}|^{2}$ and $\int \cosh^{2}W|\nabla h_{2}|^{2}$ have been removed. Furthermore,
since the two functionals only differ by boundary terms they have the same critical points.

Let $\tilde{\Psi}_{0}=(u_{0},v_{0},w_{0},\zeta^{1}_{0},\zeta^{2}_{0})$ denote the extreme Myers-Perry harmonic map (see Appendix B), and
let $\Psi_{0}=(U_{0},V_{0},W_{0},\zeta^{1}_{0},\zeta^{2}_{0})$ be the associated renormalized map with $u_{0}=U_{0}+h_{1}$, $v_{0}=V_{0}+h_{2}$, and $w_{0}=W_{0}$. Therefore,
$\Psi_{0}$ is a critical point of $\mathcal{I}$. The purpose of this section is to show that $\Psi_{0}$ achieves the global minimum for $\mathcal{I}$.

\begin{theorem}\label{infimum}
Suppose that $\Psi=(U,V,W,\zeta^{1},\zeta^{2})$ is smooth and satisfies the asymptotics \eqref{fall1}-\eqref{fall4.1}
with $\zeta^{1}|_{\Gamma}=\zeta^{1}_{0}|_{\Gamma}$ and $\zeta^{2}|_{\Gamma}=\zeta^{2}_{0}|_{\Gamma}$, then there exists a constant $C>0$ such that
\begin{equation}\label{53}
\mathcal{I}(\Psi)-\mathcal{I}(\Psi_{0})
\geq C\left(\int_{\mathbb{R}^{3}}
\operatorname{dist}_{SL(3,\mathbb{R})/SO(3)}^{6}(\Psi,\Psi_{0})dx
\right)^{\frac{1}{3}}.
\end{equation}
\end{theorem}

The primary idea behind this result is the fact that the harmonic energy, of maps with a nonpositively curved target space, is convex along geodesic deformations. This property was exploited in \cite{schoen2013convexity} to achieve a similar result where the role of extreme Myers-Perry was played by extreme Kerr. In order to apply this strategy it is necessary to show that the reduced energy inherits convexity from the harmonic energy, and for this it is helpful to cut-off the given map data in certain regimes and paste in an extreme Myers-Perry map.
More precisely, let $\delta,\varepsilon>0$ be
small parameters and define sets $\Omega_{\delta,\varepsilon}=\{\delta< r<2/\delta;
\rho>\varepsilon\}$ and $\mathcal{A}_{\delta,\varepsilon}=B_{2/\delta}\setminus
\Omega_{\delta,\varepsilon}$, where $B_{2/\delta}$ is the ball of radius $2/\delta$ centered at the origin. Suppose that
$\Psi$ has already undergone the cut-and-paste procedure, and thus satisfies
\begin{equation}\label{54}
\operatorname{supp}(U-U_{0})\subset B_{2/\delta},\text{ }\text{ }\text{ }\text{ }\text{ }
\operatorname{supp}(V-V_{0},W-W_{0},
\zeta^{1}-\zeta^{1}_{0},\zeta^{2}-\zeta^{2}_{0})\subset \Omega_{\delta,\varepsilon}.
\end{equation}
Let $\tilde{\Psi}_{t}$, $t\in[0,1]$, be a geodesic in $SL(3,\mathbb{R})/SO(3)$ which connects
$\tilde{\Psi}_{1}=\tilde{\Psi}$ and $\tilde{\Psi}_{0}$. Then $\Psi_{t}\equiv\Psi_{0}$ outside $B_{2/\delta}$ and
$(V_{t},W_{t},\zeta^{1}_{t},\zeta^{2}_{t})\equiv(V_{0},W_{0},\zeta^{1}_{0},\zeta^{2}_{0})$
on $\mathcal{A}_{\delta,\varepsilon}$, so that in particular $U_{t}=U_{0}+t(U-U_{0})$ and $V_{t}=V_{0}$ on these
regions. This linear behavior of $U_{t}$ and constancy of $V_{t}$ (in $t$) ensures that the boundary terms of \eqref{51} do not contribute when implementing convexity of the harmonic energy. From this it follows that
\begin{equation}\label{55}
\frac{d^{2}}{dt^{2}}\mathcal{I}(\Psi_{t})
\geq 2\int_{\mathbb{R}^{3}}|\nabla\operatorname{dist}_{SL(3,\mathbb{R})/SO(3)}(\Psi,\Psi_{0})|^{2}dx.
\end{equation}
Furthermore, since $\Psi_{0}$ is a critical point
\begin{equation}\label{56}
\frac{d}{dt}\mathcal{I}(\Psi_{t})|_{t=0}=0.
\end{equation}
Therefore, the conclusion of Theorem \ref{infimum} is achieved by integrating \eqref{55} and using a Sobolev inequality. In what follows we will justify each of these steps.

In order to proceed we will record the appropriate asymptotic behavior of $\Psi$ and $\Psi_{0}$. In the statements below, it is important to keep in mind that in the relevant coordinate system, the flat metric on $\mathbb{R}^{3}$ is given by
\begin{equation}
r^{2}\left(dr^{2}+r^{2}d\theta^{2}\right)+\frac{r^{4}\sin^{2}(2\theta)}{4} d\phi^{2}
=d\rho^{2}+dz^{2}+\rho^{2} d\phi^{2},
\end{equation}
with Euclidean volume form
\begin{equation}
dx=\frac{1}{2}r^{5}\sin(2\theta)dr \wedge d\theta \wedge d\phi
=\rho d\rho \wedge dz \wedge d\phi,
\end{equation}
where the transformation between polar and cylindrical coordinates is given in \eqref{polarcoord}. Thus, for example, the norms of vectors when expressed in these polar coordinates appear to have extra fall-off as compared to the corresponding expressions in traditional polar coordinates. The motivation for using this nonstandard version of polar coordinates is related to the derivation of the mass functional \eqref{MASSFUN}. For later use, we note here that the second harmonic function of \eqref{harmonicfunction} takes a simple form when expressed in polar coordinates
\begin{equation}\label{harmonicfunction1}
h_{2}=\frac{1}{2}\log\left(\frac{\sqrt{\rho^{2}+z^{2}}-z}{
\sqrt{\rho^2+z^{2}}+z}\right)=\frac{1}{2}\log\left(\frac{1-\cos(2\theta)}{1+\cos(2\theta)}\right)
=\log(\tan\theta).
\end{equation}

When stating the asymptotics there are three regimes to analyze, namely the designated asymptotically flat end ($r\rightarrow \infty$), the nondesignated end ($r\rightarrow 0$) which is either asymptotically flat or asymptotically cylindrical, and the limit at the axis ($\rho\rightarrow 0$ with $\delta\leq r\leq 2/\delta$). A motivation for the choice of asymptotics is to have the weakest conditions which guarantee finite reduced energy, and include the decay rates of the Myers-Perry harmonic maps (both extreme and non-extreme); these properties are easily shown to be satisfied by the asymptotics below.
In what follows, $\kappa>0$ is a fixed  parameter that may take on arbitrarily small values. Let us consider the asymptotically flat end first. We require that as $r\rightarrow\infty$ the following decay occurs
\begin{equation}\label{fall1}
U=O(r^{-1-\kappa}),\text{ }\text{ }\text{ }\text{ }V=O(r^{-1-\kappa}),\text{ }\text{ }\text{ }\text{ }W=\sqrt{\rho} O(r^{-2-\kappa}),
\end{equation}
\begin{equation}\label{fall2}
|\nabla U|=O(r^{-3-\kappa}),\text{ }\text{ }\text{ }\text{ }|\nabla V|=O(r^{-3-\kappa}),\text{ }\text{ }\text{ }\text{ }|\nabla W|=\rho^{-\frac{1}{2}} O(r^{-2-\kappa}),
\end{equation}
\begin{equation}\label{fall3}
|\nabla \zeta^{1}|=\rho\sqrt{\sin\theta} O(r^{-2-\kappa}),\text{ }\text{ }\text{ }\text{ }\text{ }\text{ }\text{ }\text{ }|\nabla \zeta^{2}|=\rho\sqrt{\cos\theta} O(r^{-2-\kappa}).
\end{equation}
Next consider asymptotics in the nondesignated end, which are broken up into two cases. As $r\rightarrow 0$ we require in the asymptotically flat case that
\begin{equation}\label{fall1.0}
U=-2\log r+O(1),\text{ }\text{ }\text{ }\text{ }V=O(1),\text{ }\text{ }\text{ }\text{ } W=\sqrt{\rho}O(r^{-1}),
\end{equation}
\begin{equation}\label{fall2.0}
|\nabla U|=O(r^{-2}),\text{ }\text{ }\text{ }\text{ }|\nabla V|=O(r^{-2}),\text{ }\text{ }\text{ }\text{ }|\nabla W|=\rho^{-\frac{1}{2}} O(r^{-1}),
\end{equation}
\begin{equation}\label{fall4.0}
|\nabla \zeta^{1}|=\rho\sqrt{\sin\theta} O(r^{-8+\kappa}),\text{ }\text{ }\text{ }\text{ }\text{ }\text{ }\text{ }\text{ }|\nabla \zeta^{2}|=\rho\sqrt{\cos\theta} O(r^{-8+\kappa}),
\end{equation}
and in the asymptotically cylindrical case that
\begin{equation}\label{fall5.0}
U=-\log r+O(1),\text{ }\text{ }\text{ }\text{ }V=O(1),\text{ }\text{ }\text{ }\text{ } W=\sqrt{\rho}O(r^{-1}),
\end{equation}
\begin{equation}\label{fall6.0}
|\nabla U|=O(r^{-2}),\text{ }\text{ }\text{ }\text{ }|\nabla V|=O(r^{-2}),\text{ }\text{ }\text{ }\text{ }|\nabla W|=\rho^{-\frac{1}{2}} O(r^{-1}),
\end{equation}
\begin{equation}\label{fall7.0}
|\nabla \zeta^{1}|=\rho\sqrt{\sin\theta} O(r^{-5+\kappa}),\text{ }\text{ }\text{ }\text{ }\text{ }\text{ }\text{ }\text{ }|\nabla \zeta^{2}|=\rho\sqrt{\cos\theta} O(r^{-5+\kappa}).
\end{equation}
Furthermore, the near axis asymptotics as $\rho\rightarrow 0$, $\delta\leq r\leq 2/\delta$ are required to satisfy
\begin{equation}\label{fall1.1}
U=O(1),\text{ }\text{ }\text{ }\text{ }V=O(1),\text{ }\text{ }\text{ }\text{ }W=O(\rho^{\frac{1}{2}}),
\end{equation}
\begin{equation}\label{fall2.1}
|\nabla U|=O(1),\text{ }\text{ }\text{ }\text{ }|\nabla V|=O(1),\text{ }\text{ }\text{ }\text{ }|\nabla W|=O(\rho^{-\frac{1}{2}}),
\end{equation}
\begin{equation}\label{fall4.1}
|\nabla \zeta^{1}|=\sqrt{\sin\theta}O(\rho),\text{ }\text{ }\text{ }\text{ }\text{ }\text{ }\text{ }\text{ }|\nabla \zeta^{2}|=\sqrt{\cos\theta}O(\rho).
\end{equation}

We will also have need of precise asymptotics for the extreme Myers-Perry data $\Psi_{0}$ which are derived in Appendix B. In the designated asymptotically flat end as $r\rightarrow\infty$ we have
\begin{equation}\label{fall4}
U_{0}=O(r^{-2}),\text{ }\text{ }\text{ }\text{ }V_{0}=O(r^{-2}),\text{ }\text{ }\text{ }\text{ }W_{0}=\rho O(r^{-6}),
\end{equation}
\begin{equation}\label{fall5}
|\nabla U_{0}|=O(r^{-4}),\text{ }\text{ }\text{ }\text{ }|\nabla V_{0}|=O(r^{-4}),\text{ }\text{ }\text{ }\text{ }|\nabla W_{0}|=O(r^{-6}),
\end{equation}
\begin{equation}\label{fall6}
|\nabla \zeta_{0}^{1}|=\rho\sin^{2}\theta O(r^{-4}),\text{ }\text{ }\text{ }\text{ }\text{ }\text{ }\text{ }\text{ }|\nabla \zeta_{0}^{2}|=\rho\cos^{2}\theta O(r^{-4}).
\end{equation}
In the nondesignated end as $r\rightarrow 0$ the following asymptotics are present
\begin{equation}\label{fall8.0}
U_{0}=-\log r+O(1),\text{ }\text{ }\text{ }\text{ }V_{0}=O(1),\text{ }\text{ }\text{ }\text{ }W_{0}=\rho O(r^{-2}),
\end{equation}
\begin{equation}\label{fall9.0}
|\nabla U_{0}|=O(r^{-2}),\text{ }\text{ }\text{ }\text{ }|\nabla V_{0}|=O(r^{-2}),\text{ }\text{ }\text{ }\text{ }|\nabla W_{0}|=O(r^{-2}),
\end{equation}
\begin{equation}\label{fall10.0}
|\nabla \zeta_{0}^{1}|=\rho\sin^{2}\theta O(r^{-4}),\text{ }\text{ }\text{ }\text{ }\text{ }\text{ }\text{ }\text{ }|\nabla \zeta_{0}^{2}|=\rho\cos^{2}\theta O(r^{-4}).
\end{equation}
Moreover, the near axis asymptotics as $\rho\rightarrow 0$, $\delta\leq r\leq 2/\delta$ are given by
\begin{equation}\label{fall8.1}
U_{0}=O(1),\text{ }\text{ }\text{ }\text{ }V_{0}=O(1),\text{ }\text{ }\text{ }\text{ } W_{0}=O(\rho),
\end{equation}
\begin{equation}\label{fall9.1}
|\nabla U_{0}|=O(1),\text{ }\text{ }\text{ }\text{ }|\nabla V_{0}|=O(1),\text{ }\text{ }\text{ }\text{ }|\nabla W_{0}|=O(1),
\end{equation}
\begin{equation}\label{fall10.1}
|\nabla \zeta_{0}^{1}|=\sin^{2}\theta O(\rho),\text{ }\text{ }\text{ }\text{ }\text{ }\text{ }\text{ }\text{ }|\nabla \zeta_{0}^{2}|=\cos^{2}\theta O(\rho).
\end{equation}

The first task needed to carry out the proof of Theorem \ref{infimum} as outlined above, is to first show that it is possible to
approximate $\mathcal{I}(\Psi)$ by replacing $\Psi$ with a map that satisfies \eqref{54}. This may be achieved as in
\cite{schoen2013convexity} with a three step cut and paste argument. Define cut-off functions, which only take values in the interval $[0,1]$, by
\begin{equation}\label{65}
\overline{\varphi}_{\delta}=\begin{cases}
1 & \text{ if $r\leq\frac{1}{\delta}$,} \\
|\nabla\overline{\varphi}_{\delta}|\leq 2\delta^2 &
\text{ if $\frac{1}{\delta}<r<\frac{2}{\delta}$,} \\
0 & \text{ if $r\geq\frac{2}{\delta}$,} \\
\end{cases}
\end{equation}
\begin{equation}\label{66}
\varphi_{\delta}=\begin{cases}
0 & \text{ if $r\leq\delta$,} \\
|\nabla\varphi_{\delta}|\leq \frac{2}{\delta^{2}} &
\text{ if $\delta<r<2\delta$,} \\
1 & \text{ if $r\geq2\delta$,} \\
\end{cases}
\end{equation}
and
\begin{equation}\label{67}
\phi_{\varepsilon}=\begin{cases}
0 & \text{ if $\rho\leq\varepsilon$,} \\
\frac{\log(\rho/\varepsilon)}{\log(\sqrt{\varepsilon}/
\varepsilon)} &
\text{ if $\varepsilon<\rho<\sqrt{\varepsilon}$,} \\
1 & \text{ if $\rho\geq\sqrt{\varepsilon}$.} \\
\end{cases}
\end{equation}
Let
\begin{equation}\label{68}
\overline{F}_{\delta}(\Psi)=\Psi_{0}
+\overline{\varphi}_{\delta}(\Psi-\Psi_{0})=:
(\overline{U}_{\delta},\overline{V}_{\delta},\overline{W}_{\delta}
,\overline{\zeta}^{1}_{\delta},\overline{\zeta}^{2}_{\delta}),
\end{equation}
so that $\overline{F}_{\delta}(\Psi)=\Psi_{0}$ on $\mathbb{R}^{3}\setminus B_{2/\delta}$.

\begin{lemma}\label{cutandpaste1}
$\lim_{\delta\rightarrow 0}\mathcal{I}(\overline{F}_{\delta}(\Psi))=\mathcal{I}(\Psi).$
\end{lemma}

\begin{proof}
Write
\begin{equation}\label{69}
\mathcal{I}(\overline{F}_{\delta}(\Psi))
=\mathcal{I}_{r\leq\frac{1}{\delta}}(\overline{F}_{\delta}(\Psi))
+\mathcal{I}_{\frac{1}{\delta}< r<\frac{2}{\delta}}(\overline{F}_{\delta}(\Psi))
+\mathcal{I}_{r\geq\frac{2}{\delta}}(\overline{F}_{\delta}(\Psi)),
\end{equation}
and observe that $\mathcal{I}_{r\leq\frac{1}{\delta}}(\overline{F}_{\delta}(\Psi))
\rightarrow \mathcal{I}(\Psi)$ by the dominated convergence theorem (DCT). Moreover, since $\Psi_{0}$ has finite reduced
energy $\mathcal{I}_{r\geq\frac{2}{\delta}}(\overline{F}_{\delta}(\Psi))\rightarrow 0$. Now write
\begin{align}\label{70}
\begin{split}
\mathcal{I}_{\frac{1}{\delta}< r<\frac{2}{\delta}}(\overline{F}_{\delta}(\Psi))
=&\underbrace{\int_{\frac{1}{\delta}<r<\frac{2}{\delta}}
12|\nabla \overline{U}_{\delta}|^{2}}_{I_{1}}
+\underbrace{\int_{\frac{1}{\delta}<r<\frac{2}{\delta}}
|\nabla\overline{V}_{\delta}|^{2}}_{I_{2}}
+\underbrace{\int_{\frac{1}{\delta}<r<\frac{2}{\delta}}
|\nabla\overline{W}_{\delta}|^{2}}_{I_{3}}\\
&+\underbrace{\int_{\frac{1}{\delta}<r<\frac{2}{\delta}}
\sinh^{2}\overline{W}_{\delta}|\nabla(\overline{V}_{\delta}+h_{2})|^{2}}_{I_{4}}
+\underbrace{\int_{\frac{1}{\delta}<r<\frac{2}{\delta}}
\frac{\cos\theta}{\rho^{3}\sin\theta}
\frac{e^{-\overline{V}_{\delta}-6\overline{U}_{\delta}}}{\cosh\overline{W}_{\delta}}
|\nabla\overline{\zeta}^{1}_{\delta}|^{2}}_{I_{5}}\\
&+\underbrace{\int_{\frac{1}{\delta}<r<\frac{2}{\delta}}
\frac{\sin\theta}{\rho^{3}\cos\theta}
e^{\overline{V}_{\delta}-6\overline{U}_{\delta}}\cosh\overline{W}_{\delta}
|\nabla\overline{\zeta}^{2}_{\delta}-e^{-\overline{V}_{\delta}}\cot\theta \tanh\overline{W}_{\delta}
\nabla\overline{\zeta}^{1}_{\delta}|^{2}}_{I_{6}}.
\end{split}
\end{align}
We have
\begin{equation}\label{71}
I_{1}\leq C\int_{0}^{\frac{\pi}{2}}\int_{\frac{1}{\delta}}^{\frac{2}{\delta}}
\left(\underbrace{|\nabla U|^{2}}_{O(r^{-6-2\kappa})}+\underbrace{|\nabla U_{0}|^{2}}_{O(r^{-8})}
+\underbrace{(U-U_{0})^{2}}_{O(r^{-2-2\kappa})}
\underbrace{|\nabla\overline{\varphi}_{\delta}|^{2}}_{ O(\delta^{4})}\right)r^{5}\sin(2\theta)dr d\theta\rightarrow 0.
\end{equation}
Moreover, a similar computation shows that $I_{2}\rightarrow 0$ and $I_{3}\rightarrow 0$.

Next observe that since
\begin{equation}\label{001}
\sinh\overline{W}_{\delta}=\sqrt{\rho}
O(r^{-2-\kappa}),
\end{equation}
we have
\begin{equation}
I_{4}\leq \int_{0}^{\frac{\pi}{2}}\int_{\frac{1}{\delta}}^{\frac{2}{\delta}}
\rho O(r^{-4-2\kappa})\left(\underbrace{|\nabla V|^{2}}_{O(r^{-6-2\kappa})}+
\underbrace{|\nabla V_{0}|^{2}}_{O(r^{-8})}
+\underbrace{|\nabla h_{2}|^{2}}_{O(\rho^{-2})}
+\underbrace{(V-V_{0})^{2}}_{O(r^{-2-2\kappa})}\underbrace{|\nabla\overline{\varphi}_{\delta}|^{2}}_{ O(\delta^{4})}\right)r^{5}\sin(2\theta)dr d\theta\rightarrow 0.
\end{equation}

In order to estimate the 5th integral, note that \eqref{fall3} and \eqref{fall6} combined with the fact that $(\zeta^{i}-\zeta_{0}^{i})|_{\Gamma}=0$, yields the following estimate for $r\in[\frac{1}{\delta},\frac{2}{\delta}]$ and $i=1,2$:
\begin{equation}\label{002}
|(\zeta^{i}-\zeta^{i}_{0})(\rho,z,\phi)|
\leq\int_{0}^{\rho}|\partial_{\rho}(\zeta^{i}-\zeta^{i}_{0})(\tilde{\rho},z,\phi)|d\tilde{\rho}
=\rho^{2}O(r^{-2-\kappa}).
\end{equation}
It follows that
\begin{equation}
I_{5}\leq C\int_{0}^{\frac{\pi}{2}}\int_{\frac{1}{\delta}}^{\frac{2}{\delta}}
\frac{\cos\theta}{\rho^{3}\sin\theta}\left(\underbrace{|\nabla \zeta^{1}|^{2}}_{\rho^{2}O(r^{-4-2\kappa})\sin\theta}+
\underbrace{|\nabla \zeta^{1}_{0}|^{2}}_{\rho^{2}O(r^{-8})\sin^{4}\theta}
+\underbrace{(\zeta^{1}-\zeta^{1}_{0})^{2}}
_{\rho^{4}O(r^{-4-2\kappa})}\underbrace{|\nabla\overline{\varphi}_{\delta}|^{2}}_{ O(\delta^{4})}\right)r^{5}\sin(2\theta)dr d\theta,
\end{equation}
which converges to zero.

Lastly, consider the 6th integral. Use \eqref{001} and \eqref{002} to find
\begin{align}
\begin{split}
I_{6}\leq& C\int_{0}^{\frac{\pi}{2}}\int_{\frac{1}{\delta}}^{\frac{2}{\delta}}
\frac{\sin\theta}{\rho^{3}\cos\theta}\left(\underbrace{|\nabla \zeta^{2}|^{2}}_{ \rho^{2}O(r^{-4-2\kappa})\cos\theta}+
\underbrace{|\nabla \zeta^{2}_{0}|^{2}}_{\rho^{2}O(r^{-8})\cos^{4}\theta}
+\underbrace{(\zeta^{2}-\zeta^{2}_{0})^{2}}_{ \rho^{4}O(r^{-4-2\kappa})}\underbrace{|\nabla\overline{\varphi}_{\delta}|^{2}}_{ O(\delta^{4})}\right)r^{5}\sin(2\theta)dr d\theta\\
&+C\int_{0}^{\frac{\pi}{2}}\int_{\frac{1}{\delta}}^{\frac{2}{\delta}}
\frac{r^{-4-2\kappa}\cos\theta}{\rho^{2}\sin\theta}\left(\underbrace{|\nabla \zeta^{1}|^{2}}_{ \rho^{2}O(r^{-4-2\kappa})\sin\theta}+
\underbrace{|\nabla \zeta^{1}_{0}|^{2}}_{\rho^{2}O(r^{-8})\sin^{4}\theta}
+\underbrace{(\zeta^{1}-\zeta^{1}_{0})^{2}}_{ \rho^{4}O(r^{-4-2\kappa})}\underbrace{|\nabla\overline{\varphi}_{\delta}|^{2}}_{ O(\delta^{4})}\right)r^{5}\sin(2\theta)dr d\theta.
\end{split}
\end{align}
This clearly also converges to zero.
\end{proof}

Consider now small balls centered at the origin. Let
\begin{equation}\label{68}
F_{\delta}(\Psi)=
(U,V_{\delta},W_{\delta}
,\zeta^{1}_{\delta},\zeta^{2}_{\delta}),
\end{equation}
where
\begin{equation}
(V_{\delta},W_{\delta}
,\zeta^{1}_{\delta},\zeta^{2}_{\delta})=(V_0,W_0,\zeta^{1}_0,\zeta^{2}_0)
+\varphi_{\delta}(V-V_0,W-W_0,\zeta^{1}-\zeta^{1}_0,\zeta^{2}-\zeta^{2}_0),
\end{equation}
so that $F_{\delta}(\Psi)=\Psi_{0}$ up to the first component on $B_{\delta}$.

\begin{lemma}\label{cutandpaste2}
$\lim_{\delta\rightarrow 0}\mathcal{I}(F_{\delta}(\Psi))=\mathcal{I}(\Psi).$ This also holds if $\Psi\equiv\Psi_{0}$ outside of $B_{2/\delta}$.
\end{lemma}

\begin{proof}
Write
\begin{equation}\label{080}
\mathcal{I}(F_{\delta}(\Psi))
=\mathcal{I}_{r\leq\delta}(F_{\delta}(\Psi))
+\mathcal{I}_{\delta< r<2\delta}(F_{\delta}(\Psi))
+\mathcal{I}_{r\geq2\delta}(F_{\delta}(\Psi)),
\end{equation}
and observe that by the dominated convergence theorem
\begin{equation}\label{081}
\mathcal{I}_{r\geq2\delta}(F_{\delta}(\Psi))
=\mathcal{I}_{r\geq2\delta}(\Psi)
\rightarrow \mathcal{I}(\Psi).
\end{equation}
Moreover
\begin{align}
\begin{split}
\mathcal{I}_{r\leq\delta}(F_{\delta}(\Psi))=&\int_{r\leq\delta}12|\nabla U|^{2}
+|\nabla V_{0}|^{2}+|\nabla W_{0}|^{2}
+\sinh^{2}W_{0}|\nabla(V_{0}+h_{2})|^{2}+\frac{e^{-6h_{1}-6U-h_{2}-V_{0}}}{\cosh W_{0}}|\nabla \zeta_{0}^{1}|^{2}\\
&+\int_{r\leq\delta}
e^{-6h_{1}-6U+h_{2}+V_{0}}\cosh W_{0}
\left|\nabla \zeta_{0}^{2}-e^{-h_{2}-V_{0}}\tanh W_{0}\nabla \zeta_{0}^{1}\right|^{2},
\end{split}
\end{align}
where the first term on the right-hand side converges to zero again by the DCT. The remaining terms may be estimated
by the reduced energy of $\Psi_{0}$ (and hence also converge to zero), since
\begin{equation}\label{82.1}
e^{-U}\leq Ce^{-U_{0}}
\end{equation}
near the origin.

Now consider
\begin{align}\label{070}
\begin{split}
\mathcal{I}_{\delta< r<2\delta}(F_{\delta}(\Psi))
=&\underbrace{\int_{\delta< r<2\delta}
12|\nabla U|^{2}}_{I_{1}}
+\underbrace{\int_{\delta< r<2\delta}
|\nabla V_{\delta}|^{2}}_{I_{2}}
+\underbrace{\int_{\delta< r<2\delta}
|\nabla W_{\delta}|^{2}}_{I_{3}}\\
&+\underbrace{\int_{\delta< r<2\delta}
\sinh^{2}W_{\delta}|\nabla(V_{\delta}+h_{2})|^{2}}_{I_{4}}
+\underbrace{\int_{\delta< r<2\delta}
\frac{\cos\theta}{\rho^{3}\sin\theta}
\frac{e^{-V_{\delta}-6U}}{\cosh W_{\delta}}
|\nabla \zeta^{1}_{\delta}|^{2}}_{I_{5}}\\
&+\underbrace{\int_{\delta< r<2\delta}
\frac{\sin\theta}{\rho^{3}\cos\theta}
e^{V_{\delta}-6U}\cosh W_{\delta}
|\nabla \zeta^{2}_{\delta}-e^{-V_{\delta}}\cot\theta \tanh W_{\delta}
\nabla \zeta^{1}_{\delta}|^{2}}_{I_{6}}.
\end{split}
\end{align}
Notice that
\begin{equation}\label{0071}
I_{1}=24\pi\int_{0}^{\frac{\pi}{2}}\int_{\delta}^{2\delta}
\underbrace{|\nabla U|^{2}}_{O(r^{-4})}r^{5}\sin(2\theta)dr d\theta\rightarrow 0.
\end{equation}
Also
\begin{equation}\label{00072}
I_{2}\leq C\int_{0}^{\frac{\pi}{2}}\int_{\delta}^{2\delta}
\left(\underbrace{|\nabla V|^{2}}_{O(r^{-4})}+\underbrace{|\nabla V_{0}|^{2}}_{O(r^{-4})}
+\underbrace{(V-V_{0})^{2}}_{O(1)}\underbrace{|\nabla\varphi_{\delta}|^{2}}_{ O(\delta^{-4})}\right)r^{5}\sin(2\theta)dr d\theta\rightarrow 0,
\end{equation}
and
\begin{equation}
I_{3}\leq C\int_{0}^{\frac{\pi}{2}}\int_{\delta}^{2\delta}
\left(\underbrace{|\nabla W|^{2}}_{\rho^{-1}O(r^{-2})}+\underbrace{|\nabla W_{0}|^{2}}_{\rho^{2}O(r^{-4})}
+\underbrace{(W-W_{0})^{2}}_{O(1)}\underbrace{|\nabla\varphi_{\delta}|^{2}}_{ O(\delta^{-4})}\right)r^{5}\sin(2\theta)dr d\theta\rightarrow 0.
\end{equation}

Now consider $I_{4}$. Since
\begin{equation}\label{0001}
\sinh W_{\delta}=\sqrt{\rho}O(r^{-1}),
\end{equation}
we have
\begin{equation}
I_{4}\leq C\int_{0}^{\frac{\pi}{2}}\int_{\delta}^{2\delta}
\rho r^{-2}\left(\underbrace{|\nabla V|^{2}}_{O(r^{-4})}+
\underbrace{|\nabla V_{0}|^{2}}_{O(r^{-4})}
+\underbrace{|\nabla h_{2}|^{2}}_{O(\rho^{-2})}
+\underbrace{(V-V_{0})^{2}}_{O(1)}\underbrace{|\nabla\varphi_{\delta}|^{2}}_{ O(\delta^{-4})}\right)r^{5}\sin(2\theta)dr d\theta\rightarrow 0.
\end{equation}

In order to estimate the 5th integral, note that \eqref{fall3} and \eqref{fall6} combined with the fact that $(\zeta^{1}-\zeta_{0}^{1})|_{\Gamma}=0$, yields the following estimate for $r\in[\delta,2\delta]$ and $i=1,2$:
\begin{equation}\label{0002}
|(\zeta^{i}-\zeta^{i}_{0})(\rho,z,\phi)|
\leq\int_{0}^{\rho}|\partial_{\rho}(\zeta^{i}-\zeta^{i}_{0})(\tilde{\rho},z,\phi)|d\tilde{\rho}
=\begin{cases}
\rho^{2}O(r^{-8+\kappa})\text{ }\text{ in the AF case,}\\
\rho^{2}O(r^{-5+\kappa})\text{ }\text{ in the AC case}.
\end{cases}
\end{equation}
It follows that in the asymptotically flat case
\begin{equation}
I_{5}\leq C\int_{0}^{\frac{\pi}{2}}\int_{\delta}^{2\delta}
\frac{r^{12}\cos\theta}{\rho^{3}\sin\theta}\left(\underbrace{|\nabla \zeta^{1}|^{2}}_{ \rho^{2}O(r^{-16+2\kappa})\sin\theta}+
\underbrace{|\nabla \zeta^{1}_{0}|^{2}}_{\rho^{2}O(r^{-8})\sin^{4}\theta}
+\underbrace{(\zeta^{1}-\zeta^{1}_{0})^{2}}_{ \rho^{4}O(r^{-16+2\kappa})}\underbrace{|\nabla\varphi_{\delta}|^{2}}_{ O(\delta^{-4})}\right)r^{5}\sin(2\theta)dr d\theta,
\end{equation}
and in the asymptotically cylindrical case
\begin{equation}
I_{5}\leq C\int_{0}^{\frac{\pi}{2}}\int_{\delta}^{2\delta}
\frac{r^{6}\cos\theta}{\rho^{3}\sin\theta}\left(\underbrace{|\nabla \zeta^{1}|^{2}}_{ \rho^{2}O(r^{-10+2\kappa})\sin\theta}+
\underbrace{|\nabla \zeta^{1}_{0}|^{2}}_{\rho^{2}O(r^{-8})\sin^{4}\theta}
+\underbrace{(\zeta^{1}-\zeta^{1}_{0})^{2}}_{ \rho^{4}O(r^{-10+2\kappa})}\underbrace{|\nabla\varphi_{\delta}|^{2}}_{ O(\delta^{-4})}\right)r^{5}\sin(2\theta)dr d\theta.
\end{equation}
These both converge to zero.

Lastly consider the 6th integral. In the asymptotically flat case we have
\begin{align}
\begin{split}
I_{6}\leq& C\int_{0}^{\frac{\pi}{2}}\int_{\delta}^{2\delta}
\frac{r^{12}\sin\theta}{\rho^{3}\cos\theta}\left(\underbrace{|\nabla \zeta^{2}|^{2}}_{ \rho^{2}O(r^{-16+2\kappa})\cos\theta}+
\underbrace{|\nabla \zeta^{2}_{0}|^{2}}_{\rho^{2}O(r^{-8})\cos^{4}\theta}
+\underbrace{(\zeta^{2}-\zeta^{2}_{0})^{2}}_{ \rho^{4}O(r^{-16+2\kappa})}\underbrace{|\nabla\varphi_{\delta}|^{2}}_{ O(\delta^{-4})}\right)r^{5}\sin(2\theta)dr d\theta\\
&+C\int_{0}^{\frac{\pi}{2}}\int_{\delta}^{2\delta}
\frac{r^{10}\cos\theta}{\rho^{2}\sin\theta}\left(\underbrace{|\nabla \zeta^{1}|^{2}}_{ \rho^{2}O(r^{-16+2\kappa})\sin\theta}+
\underbrace{|\nabla \zeta^{1}_{0}|^{2}}_{\rho^{2}O(r^{-8})\sin^{4}\theta}
+\underbrace{(\zeta^{1}-\zeta^{1}_{0})^{2}}_{ \rho^{4}O(r^{-16+2\kappa})}\underbrace{|\nabla\varphi_{\delta}|^{2}}_{ O(\delta^{-4})}\right)r^{5}\sin(2\theta)dr d\theta,
\end{split}
\end{align}
and in the asymptotically cylindrical case we have
\begin{align}
\begin{split}
I_{6}\leq& C\int_{0}^{\frac{\pi}{2}}\int_{\delta}^{2\delta}
\frac{r^{6}\sin\theta}{\rho^{3}\cos\theta}\left(\underbrace{|\nabla \zeta^{2}|^{2}}_{ \rho^{2}O(r^{-10+2\kappa})\cos\theta}+
\underbrace{|\nabla \zeta^{2}_{0}|^{2}}_{\rho^{2}O(r^{-8})\cos^{4}\theta}
+\underbrace{(\zeta^{2}-\zeta^{2}_{0})^{2}}_{ \rho^{4}O(r^{-10+2\kappa})}\underbrace{|\nabla\varphi_{\delta}|^{2}}_{ O(\delta^{-4})}\right)r^{5}\sin(2\theta)dr d\theta\\
&+C\int_{0}^{\frac{\pi}{2}}\int_{\delta}^{2\delta}
\frac{r^{4}\cos\theta}{\rho^{2}\sin\theta}\left(\underbrace{|\nabla \zeta^{1}|^{2}}_{ \rho^{2}O(r^{-10+2\kappa})\sin\theta}+
\underbrace{|\nabla \zeta^{1}_{0}|^{2}}_{\rho^{2}O(r^{-8})\sin^{4}\theta}
+\underbrace{(\zeta^{1}-\zeta^{1}_{0})^{2}}_{ \rho^{4}O(r^{-10+2\kappa})}\underbrace{|\nabla\varphi_{\delta}|^{2}}_{ O(\delta^{-4})}\right)r^{5}\sin(2\theta)dr d\theta.
\end{split}
\end{align}
Again both of these converge to zero.
\end{proof}

Consider now cylindrical regions around the axis $\Gamma$ and away from the origin given by
\begin{equation}\label{91}
\mathcal{C}_{\delta,\varepsilon}=
\{\rho\leq\varepsilon\}\cap
\{\delta\leq r\leq 2/\delta\},
\end{equation}
and
\begin{equation}\label{91.1}
\mathcal{W}_{\delta,\varepsilon}=
\{\varepsilon\leq \rho\leq\sqrt{\varepsilon}\}\cap
\{\delta\leq r\leq 2/\delta\}.
\end{equation}
Let
\begin{equation}\label{92}
G_{\varepsilon}(\Psi)=(U,V_{\varepsilon},
W_{\varepsilon},\zeta^{1}_{\varepsilon},\zeta^{2}_{\varepsilon})
\end{equation}
where
\begin{equation}\label{93}
(V_{\varepsilon},
W_{\varepsilon},\zeta^{1}_{\varepsilon},\zeta^{2}_{\varepsilon})=
(V_{0},W_{0},\zeta^{1}_{0},\zeta^{2}_{0})
+\phi_{\varepsilon}(V-V_{0},W-W_{0},\zeta^{1}-\zeta^{1}_{0},\zeta^{2}-\zeta^{2}_{0}),
\end{equation}
so that $G_{\varepsilon}(\Psi)=\Psi_{0}$ up to the first component on $\rho\leq\varepsilon$.

\begin{lemma}\label{cutandpaste3}
Fix $\delta>0$ and suppose that $\Psi\equiv\Psi_{0}$ up to the first component on $B_{\delta}$, then
$\lim_{\varepsilon\rightarrow 0}\mathcal{I}(G_{\varepsilon}(\Psi))=\mathcal{I}(\Psi)$. This also holds if
$\Psi\equiv\Psi_{0}$ outside $B_{2/\delta}$.
\end{lemma}

\begin{proof}
Write
\begin{equation}\label{94}
\mathcal{I}(G_{\varepsilon}(\Psi))
=\mathcal{I}_{\mathcal{C}_{\delta,\varepsilon}}(G_{\varepsilon}(\Psi))
+\mathcal{I}_{\mathcal{W}_{\delta,\varepsilon}}(G_{\varepsilon}(\Psi))
+\mathcal{I}_{\mathbb{R}^{3}\setminus(\mathcal{C}_{\delta,\varepsilon}
\cup\mathcal{W}_{\delta,\varepsilon})}(G_{\varepsilon}(\Psi)).
\end{equation}
Since
$\Psi\equiv\Psi_{0}$ up to the first component on $B_{\delta}$, the DCT and finite energy of $\Psi_{0}$ imply that
\begin{equation}\label{95}
\mathcal{I}_{\mathbb{R}^{3}\setminus(\mathcal{C}_{\delta,\varepsilon}
\cup\mathcal{W}_{\delta,\varepsilon})}(G_{\varepsilon}(\Psi))
\rightarrow \mathcal{I}(\Psi).
\end{equation}
Moreover
\begin{align}
\begin{split}
\mathcal{I}_{\mathcal{C}_{\delta,\varepsilon}}(G_{\varepsilon}(\Psi))=&
\int_{\mathcal{C}_{\delta,\varepsilon}}12|\nabla U|^{2}
+|\nabla V_{0}|^{2}+|\nabla W_{0}|^{2}
+\sinh^{2}W_{0}|\nabla(V_{0}+h_{2})|^{2}+\frac{e^{-6h_{1}-6U-h_{2}-V_{0}}}{\cosh W_{0}}|\nabla \zeta_{0}^{1}|^{2}
\\
&+\int_{\mathcal{C}_{\delta,\varepsilon}}
e^{-6h_{1}-6U+h_{2}+V_{0}}\cosh W_{0}
\left|\nabla \zeta_{0}^{2}-e^{-h_{2}-V_{0}}\tanh W_{0}\nabla \zeta_{0}^{1}\right|^{2},
\end{split}
\end{align}
where the first term on the right-hand side converges to zero again by the DCT. The remaining terms may be estimated
by the reduced energy of $\Psi_{0}$ (and hence also converge to zero), since
\begin{equation}
e^{-U}\leq C e^{-U_{0}}.
\end{equation}

Now observe that
\begin{align}
\begin{split}
\mathcal{I}_{\mathcal{W}_{\delta,\varepsilon}}(G_{\varepsilon}(\Psi))
=&\underbrace{\int_{\mathcal{W}_{\delta,\varepsilon}}
12|\nabla U|^{2}}_{I_{1}}
+\underbrace{\int_{\mathcal{W}_{\delta,\varepsilon}}
|\nabla V_{\varepsilon}|^{2}}_{I_{2}}
+\underbrace{\int_{\mathcal{W}_{\delta,\varepsilon}}
|\nabla W_{\varepsilon}|^{2}}_{I_{3}}\\
&+\underbrace{\int_{\mathcal{W}_{\delta,\varepsilon}}
\sinh^{2}W_{\varepsilon}|\nabla(V_{\varepsilon}+h_{2})|^{2}}_{I_{4}}
+\underbrace{\int_{\mathcal{W}_{\delta,\varepsilon}}
\frac{\cos\theta}{\rho^{3}\sin\theta}
\frac{e^{-V_{\varepsilon}-6U}}{\cosh W_{\varepsilon}}
|\nabla \zeta^{1}_{\delta}|^{2}}_{I_{5}}\\
&+\underbrace{\int_{\mathcal{W}_{\delta,\varepsilon}}
\frac{\sin\theta}{\rho^{3}\cos\theta}
e^{V_{\varepsilon}-6U}\cosh W_{\varepsilon}
|\nabla \zeta^{2}_{\varepsilon}-e^{-V_{\varepsilon}}\cot\theta \tanh W_{\varepsilon}
\nabla \zeta^{1}_{\varepsilon}|^{2}}_{I_{6}}.
\end{split}
\end{align}
We have
\begin{equation}
I_{1}\leq C\int_{\delta/2}^{3/\delta}\int_{\varepsilon}^{\sqrt{\varepsilon}}
\underbrace{|\nabla U|^{2}}_{O(1)}\rho d\rho d|z|\rightarrow 0,
\end{equation}
\begin{equation}
I_{2}\leq C\int_{\delta/2}^{3/\delta}\int_{\varepsilon}^{\sqrt{\varepsilon}}
\left(\underbrace{|\nabla V|^{2}}_{O(1)}+\underbrace{|\nabla V_{0}|^{2}}_{O(1)}
+\underbrace{(V-V_{0})^{2}}_{O(1)}\underbrace{|\nabla\phi_{\varepsilon}|^{2}}_{ O\left((\rho\log\varepsilon)^{-2}\right)}\right)\rho d\rho d|z|\rightarrow 0,
\end{equation}
and
\begin{equation}
I_{3}\leq C\int_{\delta/2}^{3/\delta}\int_{\varepsilon}^{\sqrt{\varepsilon}}
\left(\underbrace{|\nabla W|^{2}}_{O(\rho^{-1})}+\underbrace{|\nabla W_{0}|^{2}}_{O(1)}
+\underbrace{(W-W_{0})^{2}}_{O(\rho)}\underbrace{|\nabla\phi_{\varepsilon}|^{2}}_{ O\left((\rho\log\varepsilon)^{-2}\right)}\right)\rho d\rho d|z|\rightarrow 0,
\end{equation}

Now consider $I_{4}$. Since
\begin{equation}
\sinh W_{\varepsilon}=O(\rho^{\frac{1}{2}}),
\end{equation}
we find that
\begin{equation}
I_{4}\leq C\int_{\delta/2}^{3/\delta}\int_{\varepsilon}^{\sqrt{\varepsilon}}
\rho\left(\underbrace{|\nabla V|^{2}}_{O(1)}+
\underbrace{|\nabla V_{0}|^{2}}_{O(1)}
+\underbrace{|\nabla h_{2}|^{2}}_{O(\rho^{-2})}
+\underbrace{(V-V_{0})^{2}}_{O(1)}\underbrace{|\nabla\phi_{\varepsilon}|^{2}}_{ O\left((\rho\log\varepsilon)^{-2}\right)}\right)\rho d\rho d|z|\rightarrow 0.
\end{equation}

In order to estimate the 5th integral, note that the following estimate holds near the axis and away from the origin and for $i=1,2$:
\begin{equation}
|(\zeta^{i}-\zeta^{i}_{0})(\rho,z,\phi)|
\leq\int_{0}^{\rho}|\partial_{\rho}(\zeta^{i}-\zeta^{i}_{0})(\tilde{\rho},z,\phi)|d\tilde{\rho}
=O(\rho^{2}).
\end{equation}
It follows that
\begin{equation}
I_{5}\leq C\int_{\delta/2}^{3/\delta}\int_{\varepsilon}^{\sqrt{\varepsilon}}
\frac{\cos\theta}{\rho^{3}\sin\theta}\left(\underbrace{|\nabla \zeta^{1}|^{2}}_{ O(\rho^{2})\sin\theta}+
\underbrace{|\nabla \zeta^{1}_{0}|^{2}}_{O(\rho^{2})\sin^{4}\theta}
+\underbrace{(\zeta^{1}-\zeta^{1}_{0})^{2}}_{ O(\rho^{4})}\underbrace{|\nabla\phi_{\varepsilon}|^{2}}_{ O\left((\rho\log\varepsilon)^{-2}\right)}\right)\rho d\rho d|z|\rightarrow 0.
\end{equation}

Lastly consider the 6th integral. We have
\begin{align}
\begin{split}
I_{6}\leq& C\int_{\delta/2}^{3/\delta}\int_{\varepsilon}^{\sqrt{\varepsilon}}
\frac{\sin\theta}{\rho^{3}\cos\theta}\left(\underbrace{|\nabla \zeta^{2}|^{2}}_{ O(\rho^{2})\cos\theta}+
\underbrace{|\nabla \zeta^{2}_{0}|^{2}}_{O(\rho^{2})\cos^{4}\theta}
+\underbrace{(\zeta^{2}-\zeta^{2}_{0})^{2}}_{ O(\rho^{4})}\underbrace{|\nabla\phi_{\varepsilon}|^{2}}_{ O\left((\rho\log\varepsilon)^{-2}\right)}\right)\rho d\rho d|z|\\
&+C\int_{\delta/2}^{3/\delta}\int_{\varepsilon}^{\sqrt{\varepsilon}}
\frac{\cos\theta}{\rho^{2}\sin\theta}\left(\underbrace{|\nabla \zeta^{1}|^{2}}_{ O(\rho^{2})\sin\theta}+
\underbrace{|\nabla \zeta^{1}_{0}|^{2}}_{O(\rho^{2})\sin^{4}\theta}
+\underbrace{(\zeta^{1}-\zeta^{1}_{0})^{2}}_{ O(\rho^{4})}\underbrace{|\nabla\phi_{\varepsilon}|^{2}}_{ O\left((\rho\log\varepsilon)^{-2}\right)}\right)\rho d\rho d|z|,
\end{split}
\end{align}
which converges to zero.
\end{proof}

By composing the three cut and paste operations defined above, we obtain the desired replacement for $\Psi$ which satisfies \eqref{54}. Namely, let
\begin{equation}\label{105}
\Psi_{\delta,\varepsilon}
=G_{\varepsilon}\left(F_{\delta}\left(
\overline{F}_{\delta}(\Psi)\right)\right).
\end{equation}

\begin{prop}\label{proposition}
Let $\varepsilon\ll\delta\ll 1$ and suppose that $\Psi$ satisfies the hypotheses of Theorem \ref{infimum}. Then
$\Psi_{\delta,\varepsilon}$ satisfies \eqref{54} and
\begin{equation}\label{106}
\lim_{\delta\rightarrow 0}\lim_{\varepsilon\rightarrow 0}
\mathcal{I}(\Psi_{\delta,\varepsilon})=\mathcal{I}(\Psi).
\end{equation}
\end{prop}

We are now in a position to establish the main result of this section.\medskip

\noindent\textit{Proof of Theorem \ref{infimum}.} According to Proposition \ref{proposition}, $\Psi_{\delta,\varepsilon}$ satisfies
\eqref{54}. It follows that if $\tilde{\Psi}^{t}_{\delta,\varepsilon}$ is the geodesic connecting $\tilde{\Psi}_{0}$ to
$\tilde{\Psi}_{\delta,\varepsilon}$ as described at the beginning of this section, then
$U^{t}_{\delta,\varepsilon}=U_{0}+t(U_{\delta,\varepsilon}-U_{0})$ and $V^{t}_{\delta,\varepsilon}=V_{0}$ on $\mathcal{A}_{\delta,\varepsilon}$. Now write
\begin{equation}\label{107}
\frac{d^{2}}{dt^{2}}\mathcal{I}(\Psi^{t}_{\delta,\varepsilon})
=\underbrace{\frac{d^{2}}{dt^{2}}\mathcal{I}_{\Omega_{\delta,\varepsilon}}
(\Psi^{t}_{\delta,\varepsilon})}_{I_{1}}+
\underbrace{\frac{d^{2}}{dt^{2}}\mathcal{I}_{\mathcal{A}_{\delta,\varepsilon}}
(\Psi^{t}_{\delta,\varepsilon})}_{I_{2}},
\end{equation}
and observe that
\begin{align}\label{108}
\begin{split}
I_{1}=&\frac{d^{2}}{dt^{2}}E_{\Omega_{\delta,\varepsilon}}
(\tilde{\Psi}^{t}_{\delta,\varepsilon})
-\frac{d^{2}}{dt^{2}}\int_{\partial\Omega_{\delta,\varepsilon}
\cap\partial\mathcal{A}_{\delta,\varepsilon}}
12\left[h_{1}+2(U_{0}
+t(U_{\delta,\varepsilon}-U_{0}))\right]\partial_{\nu}h_{1}\\
&-\frac{d^{2}}{dt^{2}}\int_{\partial\Omega_{\delta,\varepsilon}
\cap\partial\mathcal{A}_{\delta,\varepsilon}}(h_{2}+V_{0})\partial_{\nu}h_{2}\\
\geq & 2\int_{\Omega_{\delta,\varepsilon}}
|\nabla\operatorname{dist}_{SL(3,\mathbb{R})/SO(3)}
(\Psi_{\delta,\varepsilon},\Psi_{0})|^{2}
\end{split}
\end{align}
where convexity of the harmonic energy was used in the last step, and
\begin{align}\label{109}
\begin{split}
I_{2}=&\int_{\mathcal{A}_{\delta,\varepsilon}}
24|\nabla(U_{\delta,\varepsilon}-U_{0})|^{2}
+36(U_{\delta,\varepsilon}-U_{0})^{2}
\frac{e^{-6h_{1}-6U_{\delta,\varepsilon}^{t}-h_{2}-V_{0}}}{\cosh W_{0}}|\nabla\zeta^{1}_{0}|^{2}\\
&+\int_{\mathcal{A}_{\delta,\varepsilon}}36(U_{\delta,\varepsilon}-U_{0})^{2}
e^{-6h_{1}-6U_{\delta,\varepsilon}^{t}+h_{2}+V_{0}}\cosh W_{0}|e^{-h_{2}-V_{0}}
\tanh W_{0}\nabla\zeta^{1}_{0}-\nabla\zeta_{0}^{2}|^{2}\\
\geq &2\int_{\mathcal{A}_{\delta,\varepsilon}}
|\nabla\operatorname{dist}_{SL(3,\mathbb{R})/SO(3)}
(\Psi_{\delta,\varepsilon},\Psi_{0})|^{2}
\end{split}
\end{align}
since $\operatorname{dist}_{SL(3,\mathbb{R})/SO(3)}(\Psi_{\delta,\varepsilon},\Psi_{0})
=\sqrt{12}|U_{\delta,\varepsilon}-U_{0}|$ on $\mathcal{A}_{\delta,\varepsilon}$ as the geodesic is parametrized on the interval $[0,1]$.

It remains to show that passing $\frac{d^{2}}{dt^{2}}$ into the integral in \eqref{109} is valid. For this it is
sufficient to show that each term on the right-hand side of the equality in \eqref{109} is uniformly integrable. There
is no issue with the first term since $U_{\delta,\varepsilon},U_{0}\in H^{1}(\mathbb{R}^{3})$.  Consider now the second
and third terms, and write $\mathcal{A}_{\delta,\varepsilon}=\mathcal{C}_{\delta,\varepsilon}
\cup B_{\delta}$. Uniform integrability will follow if
$(U_{\delta,\varepsilon}-U_{0})^{2}e^{-6t(U_{\delta,\varepsilon}-U_{0})}$ is uniformly bounded, since then these
terms may be estimated by the reduced energy of $\Psi_{0}$. This is clearly the case on
$\mathcal{C}_{\delta,\varepsilon}$, as $U$ and $U_{0}$ are bounded on this region. On $B_{\delta}$,
$U_{\delta,\varepsilon}-U_{0}\sim -\log r$ in the asymptotically flat case and
$U_{\delta,\varepsilon}-U_{0}\sim 1$ in the an asymptotically cylindrical case. Therefore, the desired
conclusion follows if $r^{6t}(\log r)^{2}$ is uniformly bounded, which occurs for $0<t_{0}<t\leq 1$. Since
$t_{0}>0$ is arbitrary, we conclude that \eqref{55} holds for $\Psi_{\delta,\varepsilon}$ when $t\in(0,1]$.

We will now use the Euler-Lagrange equations for $\Psi_{0}$ (Appendix B) to verify \eqref{56} for $\Psi_{\delta,\varepsilon}$. Choose $\varepsilon_{0}<\varepsilon$,
$\delta_{0}<\delta$ and write
\begin{equation}\label{110}
\frac{d}{dt}\mathcal{I}(\Psi^{t}_{\delta,\varepsilon})
=\underbrace{\frac{d}{dt}\mathcal{I}_{\Omega_{\delta_{0},\varepsilon_{0}}}(\Psi^{t}_{\delta,\varepsilon})}_{I_{3}}+
\underbrace{\frac{d}{dt}\mathcal{I}_{\mathcal{A}_{\delta_{0},\varepsilon_{0}}}(\Psi^{t}_{\delta,\varepsilon})}_{I_{4}}.
\end{equation}
Observe that the justification for passing $\frac{d}{dt}$ into the integrals, for $t\in(0,1]$, is similar to the arguments of the previous paragraph.
Then integrating by parts, using the Euler-Lagrange equations together with $\frac{d}{dt}\Psi_{\delta,\varepsilon}^{t}|_{t=0}=(U_{\delta,\varepsilon}-U_{0})\partial_{u}$
on $\mathcal{A}_{\delta_{0},\varepsilon_{0}}$, and noting that the relevant boundary integral over $\partial\Omega_{\delta_{0},\varepsilon_{0}}\cap\partial\mathcal{A}_{\delta_{0},\varepsilon_{0}}$ is equivalent to integrating over $\partial B_{\delta_{0}}\cup\partial\mathcal{C}_{\delta_{0},\varepsilon_{0}}$ yields
\begin{equation}\label{111}
I_{3}=O(t)-\int_{\partial B_{\delta_{0}}}24(U_{\delta,\varepsilon}-U_{0})\partial_{\nu}U_{0}
-\int_{\partial\mathcal{C}_{\delta_{0},\varepsilon_{0}}}
24(U_{\delta,\varepsilon}-U_{0})\partial_{\nu}U_{0}
\end{equation}
for small $t$, where $\nu$ is the unit outer normal pointing towards the designated asymptotically flat end.  Next, using that
$U^{t}_{\delta,\varepsilon}=U_{0}+t(U_{\delta,\varepsilon}-U_{0})$
and $\frac{d}{dt}V^{t}_{\delta,\varepsilon}=
\frac{d}{dt}W^{t}_{\delta,\varepsilon}
=\frac{d}{dt}\zeta^{1,t}_{\delta,\varepsilon}=\frac{d}{dt}\zeta^{2,t}_{\delta,\varepsilon}=0$ on $\mathcal{A}_{\delta_{0},\varepsilon_{0}}$ produces
\begin{align}\label{112}
\begin{split}
I_{4}=&O(t)+\int_{\mathcal{A}_{\delta_{0},\varepsilon_{0}}}
24\nabla U_{0}\cdot\nabla(U_{\delta,\varepsilon}-U_{0})
-6(U_{\delta,\varepsilon}-U_{0})
\frac{e^{-6h_{1}-6U_{\delta,\varepsilon}^{t}-h_{2}-V_{0}}}{\cosh W_{0}}|\nabla\zeta^{1}_{0}|^{2}\\
&-\int_{\mathcal{A}_{\delta_{0},\varepsilon_{0}}}
6(U_{\delta,\varepsilon}-U_{0})e^{-6h_{1}-6U_{\delta,\varepsilon}^{t}+h_{2}+V_{0}}\cosh W_{0}
|e^{-h_{2}-V_{0}}\tanh W_{0}\nabla\zeta^{1}_{0}-\nabla\zeta^{2}_{0}|^{2}.
\end{split}
\end{align}
Since $U_{\delta,\varepsilon}^{t}=U_{0}+O(t)$, we are motivated to integrate by parts in \eqref{112} and use the primary Euler-Lagrange equation for $U_{0}$ to obtain only boundary terms, which should then cancel with those in $I_{3}$ as $\mathcal{A}_{\delta_{0},\varepsilon_{0}}=B_{\delta_{0}}\cup\mathcal{C}_{\delta_{0},\varepsilon_{0}}$.
In order to carry this out, it is sufficient to check that
\begin{equation}
\left|\int_{\partial B_{\delta_{0}}}\underbrace{(U_{\delta,\varepsilon}-U)}_{O(|\log\delta_{0}|)}
\underbrace{\partial_{\nu}U_{0}}_{O(\delta^{-2}_{0})}\right|
\leq C|\log\delta_{0}|\delta_{0}^{2}\rightarrow 0 \text{ }\text{ }\text{ }\text{ as }\text{ }\text{ }\text{ }\delta_{0}\rightarrow 0,
\end{equation}
and
\begin{equation}
\left|\int_{\partial \mathcal{C}_{\delta_{0},\varepsilon_{0}}}\underbrace{(U_{\delta,\varepsilon}-U)}
_{O(1)}\underbrace{\partial_{\nu}U_{0}}_{O(1)}\right|
\leq C\varepsilon_{0}\rightarrow 0 \text{ }\text{ }\text{ }\text{ as }\text{ }\text{ }\text{ }\varepsilon_{0}\rightarrow 0.
\end{equation}
It follows that $I_{3}+I_{4}=0$ when $t=0$, and hence \eqref{56} holds for $\Psi_{\delta,\varepsilon}$.

Now integrating \eqref{55} twice and applying a Sobolev inequality produces
\begin{align}\label{117}
\begin{split}
\mathcal{I}(\Psi_{\delta,\varepsilon})-\mathcal{I}(\Psi_{0})
&\geq 2\int_{\mathbb{R}^{3}}|\nabla\operatorname{dist}_{SL(3,\mathbb{R})/SO(3)}
(\Psi_{\delta,\varepsilon},\Psi_{0})|^{2}dx\\
&\geq C\left(\int_{\mathbb{R}^{3}}
\operatorname{dist}_{SL(3,\mathbb{R})/SO(3)}^{6}(\Psi_{\delta,\varepsilon},\Psi_{0})dx
\right)^{\frac{1}{3}}.
\end{split}
\end{align}
By Proposition \ref{proposition} $\lim_{\delta\rightarrow 0}\lim_{\varepsilon\rightarrow 0}
\mathcal{I}(\Psi_{\delta,\varepsilon})=\mathcal{I}(\Psi)$, and thus in order to complete the
proof it suffices to show that the limits may be passed under the integral on the right-hand side. By the triangle
inequality, this will follow if
\begin{equation}\label{118}
\lim_{\delta\rightarrow 0}\lim_{\varepsilon\rightarrow 0}\int_{\mathbb{R}^{3}}
\operatorname{dist}_{SL(3,\mathbb{R})/SO(3)}^{6}(\Psi_{\delta,\varepsilon},\Psi)dx
=0.
\end{equation}

In order to establish \eqref{118}, observe that the
triangle inequality, together with the fact that the distance between two points in $SL(3,\mathbb{R})/SO(3)$ is not greater than the length of a coordinate line connecting them, produces
\begin{align}\label{119}
\begin{split}
&\operatorname{dist}_{SL(3,\mathbb{R})/SO(3)}(\Psi_{\delta,\varepsilon},\Psi)\\
\leq&\operatorname{dist}_{SL(3,\mathbb{R})/SO(3)}
((U_{\delta,\varepsilon},V_{\delta,\varepsilon},W_{\delta,\varepsilon},
\zeta^{1}_{\delta,\varepsilon},\zeta^{2}_{\delta,\varepsilon}),
(U,V_{\delta,\varepsilon},W_{\delta,\varepsilon},
\zeta^{1}_{\delta,\varepsilon},\zeta^{2}_{\delta,\varepsilon}))\\
&+\operatorname{dist}_{SL(3,\mathbb{R})/SO(3)}
((U,V_{\delta,\varepsilon},W_{\delta,\varepsilon},
\zeta^{1}_{\delta,\varepsilon},\zeta^{2}_{\delta,\varepsilon}),
(U,V,W_{\delta,\varepsilon},
\zeta^{1}_{\delta,\varepsilon},\zeta^{2}_{\delta,\varepsilon}))\\
&+\operatorname{dist}_{SL(3,\mathbb{R})/SO(3)}
((U,V,W_{\delta,\varepsilon},
\zeta^{1}_{\delta,\varepsilon},\zeta^{2}_{\delta,\varepsilon}),
(U,V,W,
\zeta^{1}_{\delta,\varepsilon},\zeta^{2}_{\delta,\varepsilon}))\\
&+\operatorname{dist}_{SL(3,\mathbb{R})/SO(3)}
((U,V,W,
\zeta^{1}_{\delta,\varepsilon},\zeta^{2}_{\delta,\varepsilon}),
(U,V,W,
\zeta^{1},\zeta^{2}_{\delta,\varepsilon}))\\
&+\operatorname{dist}_{SL(3,\mathbb{R})/SO(3)}
((U,V,W,
\zeta^{1},\zeta^{2}_{\delta,\varepsilon}),
(U,V,W,
\zeta^{1},\zeta^{2}))\\
\leq& C\left(|U-U_{\delta,\varepsilon}|+|V-V_{\delta,\varepsilon}|
+|W-W_{\delta,\varepsilon}|\right)\\
&+Ce^{-3U-3h_{1}}\left(e^{-\tfrac{1}{2}V-\tfrac{1}{2}h_{2}}
|\zeta^{1}-\zeta^{1}_{\delta,\varepsilon}|
+e^{\tfrac{1}{2}V+\tfrac{1}{2}h_{2}}
|\zeta^{2}-\zeta^{2}_{\delta,\varepsilon}|\right).
\end{split}
\end{align}
Notice that
\begin{equation}\label{120}
\int_{\mathbb{R}^{3}}|U-U_{\delta,\varepsilon}|^{6}dx\leq
\int_{\mathbb{R}^{3}\setminus B_{1/\delta}}\underbrace{|U-U_{0}|^{6}}_{O(r^{-6+6\kappa})}dx=O(\delta^{6\kappa})\rightarrow 0
\end{equation}
as $\delta\rightarrow 0$. Next we have
\begin{equation}\label{121}
\int_{\mathbb{R}^{3}}
|V-V_{\delta,\varepsilon}|^{6}dx
\leq C\left(\int_{\mathbb{R}^{3}\setminus B_{1/\delta}}\underbrace{|V-V_{0}|^{6}}_{O(r^{-6-6\kappa})}
+\int_{\mathcal{C}_{\delta,\sqrt{\varepsilon}}}\underbrace{|V-V_{0}|^{6}}_{O(1)\text{ as }\varepsilon\rightarrow 0}
+\int_{B_{2\delta}}
\underbrace{|V-V_{0}|^{6}}_{O(1)}\right),
\end{equation}
which converges to zero if $\varepsilon\rightarrow 0$ before $\delta\rightarrow 0$. Similarly
\begin{equation}
\int_{\mathbb{R}^{3}}
|W-W_{\delta,\varepsilon}|^{6}dx
\leq C\left(\int_{\mathbb{R}^{3}\setminus B_{1/\delta}}\underbrace{|W-W_{0}|^{6}}_{\rho^{3}O(r^{-12-6\kappa})}
+\int_{\mathcal{C}_{\delta,\sqrt{\varepsilon}}}\underbrace{|W-W_{0}|^{6}}_{O(\rho^{3})\text{ as }\varepsilon\rightarrow 0}
+\int_{B_{2\delta}}
\underbrace{|W-W_{0}|^{6}}_{\rho^{3}O(r^{-6})}\right).
\end{equation}

The last two terms on the right-hand side of \eqref{119} may each be treated in a similar fashion. Let us consider the first of these. Using the formulas \eqref{harmonicfunction} and \eqref{harmonicfunction1} yields
\begin{equation}
\int_{\mathbb{R}^{3}}e^{-18U-18h_{1}-3V-3h_{2}}|\zeta^{1}-\zeta^{1}_{\delta,\varepsilon}|^{6}dx
\leq \int_{\mathbb{R}^{3}\setminus B_{1/\delta}}
+\int_{\mathcal{C}_{\delta,\sqrt{\varepsilon}}}
+\int_{B_{2\delta}}e^{-18U-3V}\frac{\cos^{3}\theta}{\rho^{9}\sin^{3}\theta}
|\zeta^{1}-\zeta^{1}_{\delta,\varepsilon}|^{6}.
\end{equation}
Furthermore
\begin{equation}
\int_{\mathbb{R}^{3}\setminus B_{1/\delta}}
\underbrace{e^{-18U-3V}}_{O(1)}\frac{\cos^{3}\theta}{\rho^{9}\sin^{3}\theta}
\underbrace{|\zeta^{1}-\zeta^{1}_{\delta,\varepsilon}|^{6}}_{\rho^{12}O(r^{-12-6\kappa})\sin^{3}\theta }=O(\delta^{6\kappa})\rightarrow 0\text{ }\text{ }\text{ }\text{ as }\text{ }\text{ }\text{ }\delta\rightarrow 0,
\end{equation}
\begin{equation}
\int_{\mathcal{C}_{\delta,\sqrt{\varepsilon}}}
\underbrace{e^{-18U-3V}}_{O(1)}\frac{\cos^{3}\theta}{\rho^{9}\sin^{3}\theta}
\underbrace{|\zeta^{1}-\zeta^{1}_{\delta,\varepsilon}|^{6}}_{\rho^{12}\sin^{3}\theta} =O(\varepsilon^{\frac{5}{2}})\rightarrow 0\text{ }\text{ }\text{ }\text{ as }\text{ }\text{ }\text{ }\varepsilon\rightarrow 0,
\end{equation}
and in the asymptotically flat and asymptotically cylindrical cases respectively
\begin{equation}
\int_{B_{2\delta}}
\underbrace{e^{-18U-3V}}_{O(r^{36})}\frac{\cos^{3}\theta}{\rho^{9}\sin^{3}\theta}
\underbrace{|\zeta^{1}-\zeta^{1}_{\delta,\varepsilon}|^{6}}_{\rho^{12}O(r^{-48+6\kappa})\sin^{3}\theta }=O(\delta^{6\kappa})\rightarrow 0\text{ }\text{ }\text{ }\text{ as }\text{ }\text{ }\text{ }\delta\rightarrow 0,
\end{equation}
\begin{equation}
\int_{B_{2\delta}}
\underbrace{e^{-18U-3V}}_{O(r^{18})}\frac{\cos^{3}\theta}{\rho^{9}\sin^{3}\theta}
\underbrace{|\zeta^{1}-\zeta^{1}_{\delta,\varepsilon}|^{6}}_{\rho^{12}O(r^{-30+6\kappa})\sin^{3}\theta }=O(\delta^{6\kappa})\rightarrow 0\text{ }\text{ }\text{ }\text{ as }\text{ }\text{ }\text{ }\delta\rightarrow 0.
\end{equation}
It follows that \eqref{118} holds. \hfill\qedsymbol\medskip

\textit{Proof of Theorem \ref{TheoremI}.} By replacing $\eta_{(l)}$ with $-\eta_{(l)}$ if necessary, we may assume without loss of generality that $\mathcal{J}_{l}\geq 0$, $l=1,2$, so that $\mathcal{J}_{l}=|\mathcal{J}_{l}|$. If both $\mathcal{J}_{1}=\mathcal{J}_{2}=0$, then inequality \eqref{23} reduces to the positive mass theorem which holds under the current assumptions on the initial data. If only one angular momentum vanishes, say $\mathcal{J}_{1}=0$ and $\mathcal{J}_{2}\neq 0$, then we may perturb the initial data slightly to achieve $\mathcal{J}_{1}\neq 0$ and $\mathcal{J}_{2}\neq 0$ while preserving all other hypotheses of the theorem. The arguments below show that inequality \eqref{23} holds for the perturbed data, and hence also for the unperturbed data by letting the perturbation go to zero.

It remains to consider the case when $\mathcal{J}_{l}>0$, $l=1,2$. In this case there is an extreme Myers-Perry black hole solution that can serve as the model spacetime, giving rise to the harmonic map $\tilde{\Psi}_{0}$ used in the convexity arguments.
The asymptotic assumptions on the initial data $(M,g,k)$ imply that $(U,V,W,\zeta^{1},\zeta^{2})$ satisfy the asymptotics
\eqref{fall1}-\eqref{fall4.1}, see Appendix A.
Thus Theorem \ref{infimum} applies,
and the inequality \eqref{23} of Theorem \ref{TheoremI} follows from \eqref{lkjh} and \eqref{53}, after noting that
\begin{equation}\label{128}
\mathcal{M}(\Psi_{0})=\left(\frac{27\pi}{32}(\mathcal{J}_{1}+\mathcal{J}_{2})^{2}\right)^{\frac{1}{3}}.
\end{equation}

Consider now the case of equality in \eqref{23} when $\mathcal{J}_{l}>0$, $l=1,2$. As alluded to above, only in this case of nonvanishing angular momenta do we have a proper black hole spacetime arising from the extreme Myers-Perry family. If only one of the angular momenta vanish, the corresponding extreme Myers-Perry solution has a naked singularity, and such data do not satisfy the asymptotic hypotheses of the theorem. If both angular momenta vanish, then the corresponding extreme Myers-Perry data is isometric to Euclidean space minus a point, and such data again do not satisfy the hypotheses; in fact, more generally, this is true for extreme Myers-Perry data with $\mathcal{J}_{1}+\mathcal{J}_{2}=0$. Continuing with the proof in the case of nonvanishing angular momentum, observe that equality in \eqref{23} together with \eqref{lkjh} and \eqref{53} implies that
\begin{equation}\label{43}
\mu=0,\text{ }\text{ }\text{ }\text{ }\text{
}\text{ }
A^{i}_{\rho,z}=A^{i}_{z,\rho},\text{ }\text{ }\text{ }\text{ }\text{ }i=1,2,
\end{equation}
\begin{equation}\label{kequation}
k(e_{i},e_{j})=k(e_{3},e_{3})=k(e_{3},e_{4})=k(e_{4},e_{4})=0,\text{ }\text{ }\text{ }\text{ }\text{ }i,j\neq 3, 4,
\end{equation}
and
\begin{equation}\label{44}
\mathcal{M}(U,V,W,\zeta^{1},\zeta^{2})=\mathcal{M}(U_0,V_0,W_0,\zeta^{1}_0,\zeta^{2}_0).
\end{equation}
Furthermore, according to the gap bound \eqref{53}, a map which minimizes the functional $\mathcal{M}$ must coincide with
the harmonic map associated with the extreme Myers-Perry spacetime, that is
\begin{equation}\label{45}
(U,V,W,\zeta^{1},\zeta^{2})=(U_0,V_0,W_0,\zeta^{1}_0,\zeta^{2}_0).
\end{equation}
Next notice that \eqref{asdf}, \eqref{kequation}, and \eqref{45} yield
\begin{align}\label{909}
\begin{split}
R=&16\pi\mu+|k|^{2}\\
=&16\pi\mu
+\frac{e^{-8U-2\alpha+2\log r}}{2\rho^{2}}\nabla\zeta^{t}\lambda^{-1}\nabla\zeta\\
=&\frac{e^{-8U_{0}-2\alpha+2\log r}}{2\rho^{2}}\nabla\zeta_{0}^{t}\lambda^{-1}_{0}\nabla\zeta_{0}\\
=&e^{2(\alpha_{0}-\alpha)}R_{0},
\end{split}
\end{align}
where $\alpha_{0}$ and $R_{0}$ are corresponding quantities for the extreme Myers-Perry solution. On the other hand, using the scalar curvature formula \eqref{SCALAR}, together with \eqref{43} and \eqref{45} implies that
\begin{align}
\begin{split}
e^{2U+2\alpha-2\log r}R
=&-6\Delta U_{0}-2\Delta_{\rho,z}\alpha-6|\nabla U_{0}|^2
+\frac{\det\nabla\lambda_{0}}{2\rho^{2}}\\
=&e^{2U_{0}+2\alpha_{0}-2\log r}R_{0}+2\Delta_{\rho,z}(\alpha_{0}-\alpha).
\end{split}
\end{align}
It then follows from \eqref{45} and \eqref{909} that $\Delta_{\rho,z}(\alpha_{0}-\alpha)=0$. In light of the condition \eqref{cone} on the axis to avoid conical singularities, we have $(\alpha_{0}-\alpha)|_{\Gamma}=0$. Moreover $(\alpha_{0}-\alpha)\rightarrow 0$ as $r\rightarrow\infty$. Hence the maximum principle shows that $\alpha=\alpha_{0}$.

We are now in a position to show that $(M,g)$ is isometric to the canonical slice of the extreme Myers-Perry black hole.
By \eqref{43} the 1-forms $A_{\rho}^{i}d\rho+A_{z}^{i}dz$, $i=1,2$ are closed, and so there exist potentials such that
$\partial_{\rho}f^{i}=A_{\rho}^{i}$ and $\partial_{z}f^{i}=A_{z}^{i}$, $i=1,2$. Then under the change of coordinates
$\widetilde{\phi}^{i}=\phi^{i}+f^{i}(\rho,z)$, the metric takes the form
\begin{equation}\label{49}
g=\frac{e^{2U_0+2\alpha_{0}}}{2\sqrt{\rho^{2}+z^{2}}}(d\rho^{2}+dz^{2})
+e^{2U_0}(\lambda_{0})_{ij}d\widetilde{\phi}^{i}d\widetilde{\phi}^{j},
\end{equation}
which yields the desired result $g\cong g_0$. Lastly \eqref{equation}, \eqref{kequation}, \eqref{45}, and $\alpha=\alpha_{0}$ show that the tensor $k$ coincides with the extrinsic curvature of the canonical extreme Myers-Perry slice. Note that this also shows that the linear momentum vanishes $J=0$.
\hfill\qedsymbol

\section{Discussion}

In this paper we have established the mass-angular momentum inequality for 4-dimensional initial data having horizons of spherical topology, and which admit a Brill coordinate representation. There are many directions for possible generalizations. First, we strongly suspect that as in the 3-dimensional case \cite{chrusciel2008masspositivity}, the existence of Brill coordinates always occurs for data with simple topology and appropriate asymptotics. Therefore Theorem \ref{TheoremI} should hold without the Brill coordinate hypothesis. Second, it is natural to consider such inequalities for data with multiple horizons. In the 3-dimensional setting such inequalities were obtain \cite{chrusciel2008mass,khuri2015positive} in both the charged and uncharged cases, however the mass lower bound was not given explicitly. In order to carry this out in the higher dimensional setting, one would first need to construct a harmonic map to serve in the place of the Myers-Perry harmonic map $\tilde{\Psi}_{0}$. Such an existence result for a harmonic map with `multiple horizons' should be possible through an application of Weinstein's theory \cite{Weinstein}. However, the convexity arguments would be much more difficult to carry out, as such harmonic maps are not given explicitly.

Perhaps the most challenging and interesting generalization would be to allow horizons with nontrivial topology. In this situation the orbit space structure would change. In general, the 2-dimensional orbit space $M^{4}/U(1)^2$ is a
simply connected manifold with boundaries and corners \cite{Alaeeremarks2015,Hollands2008}. The boundary $\Gamma=\{\rho =0 \}$ is divided into rod intervals ${I}_s = \{\rho=0, a_{s} \leq z \leq a_{s+1}\}$, $1\leq s\leq \bar{s}+1$ where $a_1 < a_2 < \cdots < a_{\bar{s}+1}$, and on each such rod segment $\lambda$ has rank 1 or 2.  In particular, on each ${I}_s$ a certain integral linear combination of the $\eta_{(l)}$, $l=1,2$ vanishes, that is, there exists a vector $n_s^l \eta_{(l)}$, with $n_{s}^{l}\in\mathbb{Z}$, which lies in the kernel of $\lambda$, namely $\lambda_{ij} n_s^j = 0$. One may then give each rod a two component label $(n_{s}^{1},n_{s}^{2})$, indicating which linear combination vanishes. Horizons carry the label $(0,0)$, and all other rods have the property that $\lambda$ is of rank 1 while at corner points $\lambda$ is of rank 0. Moreover, asymptotic flatness implies the existence of two semi-infinite intervals ${I}_1 = \{-\infty < z < a_2\}$ and ${I}_{\bar{s}} = \{a_{\bar{s}} < z < \infty\}$ with the labels $(0,1)$ and $(1,0)$ respectively (after perhaps choosing an appropriate coordinate basis). The collection of rod intervals $I_{s}$ together with the associated labels is referred to as the orbit space data. As we have seen in this paper, the orbit space data for the extreme Myers-Perry solution consists only of the two semi-infinite rods, and the same is true for a non-extreme Myers-Perry in Brill coordinates while in Weyl coordinates it has an extra rod with label $(0,0)$ in between that represents the horizon. Consider now the black ring solution. The extreme version has three rods $I_{s}$, $s=1,2,3$ with corresponding labels $(0,1)$, $(1,0)$, and $(1,0)$. The point between $I_{2}$ and $I_{3}$ represents a cylindrical end with cross-section having topology $S^{1}\times S^{2}$. The non-extreme black ring has an extra rod between $I_{2}$ and $I_{3}$ with label $(0,0)$ to encode the horizon. When trying to establish the mass-angular momentum inequality for black holes with $S^{1}\times S^{2}$ topology, the main difficulty occurs from the fact that the orbit space structure for the model (extreme black ring) is not compatible with the orbit space structure for manifolds with two asymptotically flat ends. Thus it is not clear if a Brill coordinate description is possible, on which arbitrary initial data may be compared with the model. In particular, it is not even clear if there is a single Brill coordinate description which is compatible with both the extreme and non-extreme black ring data. On the other hand, some positive results have been obtained in the direction of a mass-angular momentum inequality for nontrivial topologies. Namely, a slight variation of the the mass functional \eqref{MASSFUN} may be derived for very general orbit space data, and it is known to be nonnegative for special classes of rod structures which include that of the extreme black ring \cite{alaee2014mass}. Ultimately, however, for nontrivial topologies it may be more appropriate to use Weyl coordinates in which the horizon is represented as a rod, instead of Brill coordinates in which the horizon is represented as a point.

\appendix
\numberwithin{equation}{section}

\section{Asymptotics}

Here we compute the asymptotics of the harmonic map data $(U,V,W,\zeta^1,\zeta^2)$ which are implied by the asymptotics of the generalized Brill data in \eqref{F1}-\eqref{F12}, and observe that they are stronger than those \eqref{fall1}-\eqref{fall4.1} which are needed to carry out the convexity arguments of Section \ref{sec4}. The asymptotics of $U$ are given directly, and those of $V$ and $W$ may be derived from the equations \eqref{VW}. Thus, it remains to compute the asymptotics for the potentials $\zeta^1$ and $\zeta^2$. Observe that
\eqref{equation} yields
\begin{equation}
|\nabla \zeta^{i}|=\left(|\partial_{\rho}\zeta^i|^{2}+|\partial_{z}\zeta^i|^{2}\right)^{\frac{1}{2}}
\leq Cr^{-1}\rho e^{4U+\alpha}\left(|k(e_{1},e_{i+2})|+|k(e_{2},e_{i+2})|\right).
\end{equation}
Furthermore, asymptotics for $k(e_{l},e_{i+2})$, $l=1,2$ may be obtained from the asymptotics of $|k|_{g}$ and $\lambda$ through the inequality
\begin{equation}
\sum_{l=1,2}\lambda^{ij}k(e_{l},e_{i+2})k(e_{l},e_{j+2})\leq |k|_{g}^{2}.
\end{equation}
In conclusion, Brill asymptotics imply the following asymptotics for the harmonic map data.
In the designated asymptotically flat end as $r\rightarrow\infty$
\begin{equation}
U=O(r^{-1-\kappa}),\text{ }\text{ }\text{ }\text{ }V=O(r^{-1-\kappa}),\text{ }\text{ }\text{ }\text{ }W=\rho O(r^{-5-\kappa}),
\end{equation}
\begin{equation}
|\nabla U|=O(r^{-3-\kappa}),\text{ }\text{ }\text{ }\text{ }|\nabla V|=O(r^{-3-\kappa}),\text{ }\text{ }\text{ }\text{ }|\nabla W|= O(r^{-5-\kappa}),
\end{equation}
\begin{equation}
|\nabla \zeta^{1}|=\rho\sin\theta O(r^{-2-\kappa}),\text{ }\text{ }\text{ }\text{ }\text{ }\text{ }\text{ }\text{ }|\nabla \zeta^{2}|=\rho\cos\theta O(r^{-2-\kappa}).
\end{equation}
As $r\rightarrow 0$ in the asymptotically flat case
\begin{equation}
U=-2\log r+O(1),\text{ }\text{ }\text{ }\text{ }V=O(r^{1+\kappa}),\text{ }\text{ }\text{ }\text{ } W=\rho O(r^{-1+\kappa}),
\end{equation}
\begin{equation}
|\nabla U|=O(r^{-2}),\text{ }\text{ }\text{ }\text{ }|\nabla V|=O(r^{-1+\kappa}),\text{ }\text{ }\text{ }\text{ }|\nabla W|=O(r^{-1+\kappa}),
\end{equation}
\begin{equation}
|\nabla \zeta^{1}|=\rho\sin\theta O(r^{-6+\kappa}),\text{ }\text{ }\text{ }\text{ }\text{ }\text{ }\text{ }\text{ }|\nabla \zeta^{2}|=\rho\cos\theta O(r^{-6+\kappa}),
\end{equation}
and in the asymptotically cylindrical case
\begin{equation}
U=-\log r+O(1),\text{ }\text{ }\text{ }\text{ }V=O(r^{1+\kappa}),\text{ }\text{ }\text{ }\text{ } W=\rho O(r^{-2}),
\end{equation}
\begin{equation}
|\nabla U|=O(r^{-2}),\text{ }\text{ }\text{ }\text{ }|\nabla V|=O(r^{-1+\kappa}),\text{ }\text{ }\text{ }\text{ }|\nabla W|= O(r^{-2}),
\end{equation}
\begin{equation}
|\nabla \zeta^{1}|=\rho\sin\theta O(r^{-2+\kappa}),\text{ }\text{ }\text{ }\text{ }\text{ }\text{ }\text{ }\text{ }|\nabla \zeta^{2}|=\rho\cos\theta O(r^{-2+\kappa}).
\end{equation}
Lastly, as $\rho\rightarrow 0$ with $\delta\leq r\leq 2/\delta$ we have
\begin{equation}
U=O(1),\text{ }\text{ }\text{ }\text{ }V=O(1),\text{ }\text{ }\text{ }\text{ }W=O(\rho^{\frac{1}{2}}),
\end{equation}
\begin{equation}
|\nabla U|=O(1),\text{ }\text{ }\text{ }\text{ }|\nabla V|=O(1),\text{ }\text{ }\text{ }\text{ }|\nabla W|=O(\rho^{-\frac{1}{2}}),
\end{equation}
\begin{equation}
|\nabla \zeta^{1}|=\sin\theta O(\rho),\text{ }\text{ }\text{ }\text{ }\text{ }\text{ }\text{ }\text{ }|\nabla \zeta^{2}|=\cos\theta O(\rho).
\end{equation}

\section{The Extreme Myers-Perry Harmonic Map} \label{AppMP}

The Myers-Perry black holes \cite{Myers1986} are solutions to the vacuum Einstein equations in all dimensions greater than four, and have horizons of spherical topology. They are considered to be the natural generalization to higher dimensions of the 4-dimensional Kerr black holes. In coordinates analogous to those of Boyer-Lindquist used for the Kerr solution, the Myers-Perry metric takes the form
\begin{align}
\begin{split}
&-dt^2+\frac{\mathfrak{m}}{\Sigma}\left(dt+a\sin^2\theta d\phi^1+b\cos^2\theta d\phi^2\right)^2+\frac{\tilde{r}^2\Sigma}{\Delta}d\tilde{r}^2\\
&+ \Sigma d\theta^2+\left(\tilde{r}^2+a^2\right)\sin^2\theta (d\phi^1)^2
+\left(\tilde{r}^2+b^2\right)\cos^2\theta (d\phi^2)^2,
\end{split}
\end{align}
where
\begin{gather}
\Sigma=\tilde{r}^2+b^2\sin^2\theta+a^2\cos^2\theta,\qquad \Delta=\left(\tilde{r}^2+a^2\right)\left(\tilde{r}^2+b^2\right)-\mathfrak{m} \tilde{r}^2.
\end{gather}
This family of solutions is parameterized by $(\mathfrak{m},a,b)$ which give rise to the mass and angular momenta through the formulae
\begin{gather}
m=\frac{3}{8}\pi\mathfrak{m},\qquad \mathcal{J}_{1}=\frac{2}{3} ma,\qquad \mathcal{J}_{2}=\frac{2}{3} mb;
\end{gather}
the black hole is referred to as extreme if $\mathfrak{m}=(a+b)^2$.
Note that this spacetime has the orthogonally transitive isometry group $\mathbb{R}\times U(1)^2$, where $\mathbb{R}$ gives the time translation symmetry and $U(1)^2$ is the rotational symmetry generated by $\partial_{\phi^1}$ and $\partial_{\phi^2}$. Here $(\tilde{r},\theta)$ parameterize the 2-dimensional surfaces orthogonal to the orbits of the isometry group. The horizons of this black hole are located at the roots of $\Delta$, namely
\begin{equation}
\tilde{r}_{\pm}=\pm\sqrt{\frac{\mathfrak{m}-a^2-b^2
+\sqrt{\left(\mathfrak{m}-a^2-b^2\right)^2-4a^2b^2}}{2}},
\end{equation}
and the singularities of this metric for nonvanishing $a$ and $b$ with $|a|\neq |b|$ are located at the roots of $\Sigma$.
We will restrict attention to the exterior region
$\tilde{r}>\tilde{r}_{+}$, with the other variables having ranges $0 < \theta < \pi/2$ and $0<\phi^1,\phi^2<2\pi$.

Consider now the metric on a constant time slice. In the exterior region, this may be put into Brill form by defining a new radial coordinate $r$:
\begin{equation}\label{appendix1}
\tilde{r}^2=r^2+\frac{1}{2}\left(\mathfrak{m}-a^2-b^2\right)
+\frac{\mathfrak{m}\left(\mathfrak{m}-2a^2-2b^2\right)+(a^2-b^2)^2}{16r^2},\text{ }\text{ }\text{ }\text{ }\text{ }\mathfrak{m}\neq (a+b)^2,
\end{equation}
\begin{equation}
\tilde{r}^2=r^2+ab,\text{ }\text{ }\text{ }\text{ }\text{ }\mathfrak{m}= (a+b)^2.
\end{equation}
Observe that the new coordinate is defined on the interval $(0,\infty)$, and a critical point for the right-hand side of \eqref{appendix1} occurs at the horizon, so that two isometric copies of the outer region are encoded on this
interval. The coordinates $(r,\theta,\phi^1,\phi^2)$ then give a (polar) Brill coordinate system, where the spatial metric takes the form
\begin{eqnarray}
g=\frac{\Sigma}{r^2}dr^2
+\Sigma d\theta^2+\Lambda_{ij}d\phi^i d\phi^j,
\end{eqnarray}
with
\begin{equation}
\Lambda_{11}=\frac{a^2\mathfrak{m}}{\Sigma}\sin^4\theta+(\tilde{r}^2+a^2)\sin^2\theta,\quad \Lambda_{12}=\frac{ab\mathfrak{m}}{\Sigma}\sin^2\theta\cos^2\theta,\quad
\Lambda_{22}=\frac{b^2\mathfrak{m}}{\Sigma}\cos^4\theta+(\tilde{r}^2+b^2)\cos^2\theta.
\end{equation}
Cylindrical Brill coordinates may be obtained via the usual transformation $\rho=\frac{1}{2}r^{2}\sin(2\theta)$, $z=\frac{1}{2}r^{2}\cos(2\theta)$, so that the metric is given by
\begin{equation}
g=\frac{e^{2U+2\alpha}}{2\sqrt{\rho^2+z^2}}(d\rho^2+dz^2)
+e^{2U}\lambda_{ij}d\phi^i d\phi^j,
\end{equation}
where
\begin{equation}\label{htildeMP}
e^{2U}=\frac{\sqrt{\det\Lambda}}{\rho}, \qquad e^{2\alpha}=\frac{\rho \Sigma}{r^2\sqrt{\det\Lambda}}, \qquad \lambda_{ij}=\frac{\rho}{\sqrt{\det\Lambda}}\Lambda_{ij}.
\end{equation}
From this we may compute the harmonic map data $(U,V,W)$ with the help of \eqref{VW}. Moreover, the twist potentials are given in the non-extreme case by
\begin{align}
\begin{split}
\zeta^1=&\frac{\left[\mathcal{C}_{1}^{2}+256r^{4}\Sigma(a^2-b^2)
\cos^2\theta\right]\left(\mathcal{C}_{1}-16r^2\left(a^2-b^2\right)\right)\mathfrak{m} a}{16^3r^6\Sigma (a^2-b^2)^2}-\frac{\mathcal{C}_{1}^{2}-32r^4\mathcal{C}_2}{256r^4(a^2-b^2)^2}\mathfrak{m} a,\\
\zeta^2=&-\frac{\left[\mathcal{C}_{1}
\left(\mathcal{C}_{1}-32r^2\left(a^2-b^2\right)\right)+256r^4(a^2-b^2)
(\Sigma\cos^2\theta+(a^2-b^2))\right]\mathcal{C}_{1}\mathfrak{m} b}{4096r^6\Sigma (a^2-b^2)^2}\\
&+\frac{\mathcal{C}_{1}^{2}-16r^2\mathcal{C}_{1}\left(a^2-b^2\right)
+32r^4\mathcal{C}_{3}}{256r^4(a^2-b^2)^2}\mathfrak{m} b,
\end{split}
\end{align}
where
\begin{equation}
\mathcal{C}_{1}=16r^4+8(\mathfrak{m}+a^2-b^2)r^2
+\left(\mathfrak{m}-\left(a-b\right)^2\right)\left(\mathfrak{m}-\left(a+b\right)^2\right),
\end{equation}
\begin{equation}
\mathcal{C}_2=3(a^2-b^2)^2+\mathfrak{m}\left(3\mathfrak{m}-6b^2+2a^2\right),\text{ }\text{ }\text{ }\text{ }\text{ } \mathcal{C}_{3}=(a^2-b^2)^2+\mathfrak{m}\left(2a^2+2b^2-3\mathfrak{m}\right),
\end{equation}
and in the extreme case by
\begin{align}
\begin{split}
\zeta^1_0=&\frac{a(a^2-b^2)(r^2+ab+b^2)\cos^2\theta
-r^2a(2a^2+2ab+r^2)}{(a-b)^2}+\frac{a(r^2+ab+a^2)^2(r^2+ab+b^2)}{\Sigma(a-b)^2},\\ \zeta^2_0=&\frac{br^2((a+b)^2+r^2)-b(a^2-b^2)(r^2+ab+a^2)\cos^2\theta}{(a-b)^2}
-\frac{b(r^2+ab+a^2)(r^2+ab+b^2)^2}{\Sigma(a-b)^2}.
\end{split}
\end{align}
The asymptotics of the non-extreme Myers-Perry data are then as follows. In the designated asymptotically flat end as $r\rightarrow\infty$ we have
\begin{equation}
U=O(r^{-2}),\text{ }\text{ }\text{ }\text{ }V=O(r^{-2}),\text{ }\text{ }\text{ }\text{ }W=\rho O(r^{-6}),
\end{equation}
\begin{equation}
|\nabla U|=O(r^{-4}),\text{ }\text{ }\text{ }\text{ }|\nabla V|=O(r^{-4}),\text{ }\text{ }\text{ }\text{ }|\nabla W|=O(r^{-6}),
\end{equation}
\begin{equation}
|\nabla \zeta^{1}|=\rho\sin^{2}\theta O(r^{-4}),\text{ }\text{ }\text{ }\text{ }\text{ }\text{ }\text{ }\text{ }|\nabla \zeta^{2}|=\rho\cos^{2}\theta O(r^{-4}).
\end{equation}
In the nondesignated (asymptotically flat) end as $r\rightarrow 0$ it holds that
\begin{equation}
U=-2\log r+O(1),\text{ }\text{ }\text{ }\text{ }V=O(r^2),\text{ }\text{ }\text{ }\text{ }W=\rho O(r^{2}),
\end{equation}
\begin{equation}
|\nabla U|=O(r^{-2}),\text{ }\text{ }\text{ }\text{ }|\nabla V|=O(1),\text{ }\text{ }\text{ }\text{ }|\nabla W|=O(r^{2}),
\end{equation}
\begin{equation}
|\nabla \zeta^{1}|=\rho\sin^{2}\theta O(r^{-4}),\text{ }\text{ }\text{ }\text{ }\text{ }\text{ }\text{ }\text{ }|\nabla \zeta^{2}|=\rho\cos^{2}\theta O(r^{-4}).
\end{equation}
Furthermore, the near axis asymptotics as $\rho\rightarrow 0$, $\delta\leq r\leq 2/\delta$ are given by
\begin{equation}
U=O(1),\text{ }\text{ }\text{ }\text{ }V=O(1),\text{ }\text{ }\text{ }\text{ } W=O(\rho),
\end{equation}
\begin{equation}
|\nabla U|=O(1),\text{ }\text{ }\text{ }\text{ }|\nabla V|=O(1),\text{ }\text{ }\text{ }\text{ }|\nabla W|=O(1),
\end{equation}
\begin{equation}
|\nabla \zeta^{1}|=\sin^{2}\theta O(\rho),\text{ }\text{ }\text{ }\text{ }\text{ }\text{ }\text{ }\text{ }|\nabla \zeta^{2}|=\cos^{2}\theta O(\rho).
\end{equation}
Asymptotics in the extreme case may be computed similarly, and are recorded in \eqref{fall4}-\eqref{fall10.1}.

Lastly we note that the extreme Myers-Perry harmonic map $\tilde{\Psi}_{0}=(u_0,v_0,w_0,\zeta^{1}_{0},\zeta^{2}_{0}):\mathbb{R}^{3}\setminus\Gamma
\rightarrow SL(3,\mathbb{R})/SO(3)$ satisfies the Euler-Lagrange equations arising from the energy \eqref{energy1}, namely
\begin{align}
\begin{split}
4\Delta u+\frac{e^{-6u-v}}{\cosh w}|\nabla \zeta^{1}|^{2}+e^{-6u+v}\cosh w
\left|e^{-v}\tanh w\nabla \zeta^{1}-\nabla \zeta^{2}\right|^{2}&=0,\\
2\operatorname{div}\left(\cosh^2 w\nabla v\right)+e^{-6u-v}\cosh w|\nabla \zeta^{1}|^{2}- e^{-6u+v}\cosh w\left|\nabla\zeta^2\right|^2&=0,\\
2\Delta w-\sinh 2 w|\nabla v|^2-e^{-6u-v}\sinh w \left|\nabla\zeta^1\right|^2&\\
+2e^{-6u}\cosh w\delta_{3}(\nabla\zeta^1,\nabla\zeta^2)-e^{-6u+v}\sinh w\left|\nabla\zeta^2\right|^2&=0,\\
\operatorname{div}\left(e^{-6u-v}\cosh w\nabla \zeta^{1}-e^{-6u}\sinh w\nabla \zeta^{2}\right)&=0,\\
\operatorname{div}\left(e^{-6u}\sinh w\nabla \zeta^{1}-e^{-6u+v}\cosh w\nabla \zeta^{2}\right)&=0.
\end{split}
\end{align}

\bibliographystyle{abbrv}
\bibliography{masterfile}
\end{document}